\documentclass[a4paper]{article}
\usepackage{RRA4}
\RRNo{6496}
\usepackage{amssymb,amsmath}
\usepackage[T1]{fontenc}
\usepackage[latin1]{inputenc}

\usepackage{subfigure}
\usepackage{graphicx}
\usepackage{multirow}

\usepackage{hyperref}

\RRdate{Avril 2008}
\RRauthor{% les auteurs
Laurent Viennot\thanks{INRIA Rocquencourt, France}%
  \and
Yacine Boufkhad\thanks[liafa]{LIAFA, Paris, France}%
 \and
 Fabien Mathieu\thanks[orange]{Orange Labs, Issy-les-Moulineaux, France}
 \and
 Fabien de Montgolfier\thanksref{liafa}
 \and
 Diego Perino\thanksref{orange}
}
\authorhead{Viennot \& al.}
\RRtitle{Passage à l'échelle de services distribués de vidéos-à-la-demande}
\RRetitle{Scalable Distributed Video-on-Demand: \\ Theoretical Bounds and Practical Algorithms}
\titlehead{Scalable Distributed Video-on-Demand}
\RRnote{Supported by ANR project ``ALADDIN''.}
\RRnote{Supported by collaborative project ``MARDI II'' between INRIA and Orange Labs.}
\RRresume{Nous considérons un système de $n$ n\oe uds (\emph{boîtes}) qui hébergent un ensemble de films et cherche à diffuser jusqu'à $n$ flux vidéos simultanés. Une question qui se pose est de savoir sous quelles conditions un tel système peut passer à l'échelle tout en supportant n'importe quelle séquence de demandes. Ce problème se décompose en deux parties: la répartition initiale des vidéos dans les boîtes et l'allocation des ressources en fonction des demandes. Pour ce dernier problème, nous supposons qu'un adversaire émet des demandes de manière incrémentale.

Les principaux paramètres du problème sont : le rapport $u$ entre l'upload moyen des boîtes et le débit nécessaire à la lecture de la vidéo ; le nombre maximal $c$ de connections utilisables dans la récupération d'une vidéo ; le nombre $m$ de films distincts stockés dans le system (la taille du catalogue).

Dans un système homogène (toutes les boîtes ont les mêmes capacités), nous donnons à $c$ fixé les conditions nécessaires à la réalisation d'un système capable de passer à l'échelle.
Nous montrons en particulier que $\max\set{1+\frac{1}{c},\mun}$ est une borne inférieure pour $u$,
$\mun\ge 1$ étant le facteur de croissance maximal des ensembles de demandes d'une vidéo pendant une période de temps de l'ordre du temps d'amorce de lecture d'une vidéo (notre modèle tolère ainsi une croissance exponentielle des demandes).

Réciproquement, nous prouvons que si $u\ge\max\set{1+\frac{1}{c},\mun}$, $c\ge 2$ et si les boîtes sont fiables, alors il est possible d'avoir une taille de catalogue $m=\Omega(n)$, avec un algorithme d'allocation centralisé.

Enfin, nous proposons un algorithme d'allocation distribué, associé à un algorithme de répartition aléatoire adapté aux systèmes hétérogènes où la capacité d'upload des boîtes est proportionnelle à leur capacité de stockage. Il est alors possible de servir une séquence de $O(n)$ demandes adversariales avec forte probabilité, avec une taille de catalogue en $\Omega\left({n}/{\log n}\right)$, à la condition d'avoir $u\ge\mun+\frac{1}{c}$.
}
\RRabstract{We analyze a distributed system where $n$ nodes called
\emph{boxes} store a large set of videos and collaborate to serve
simultaneously $n$ videos or less.  We explore under which
conditions such a system can be scalable while serving any
sequence of demands.  We model this problem through a combination
of two algorithms: a video allocation algorithm and a connection
scheduling algorithm. The latter plays against an adversary that
incrementally proposes video requests.

Our main parameters are: the ratio $u$ of the average upload bandwidth
of a box to the playback rate of a video; the maximum number of
connections $c$ used for downloading a video; the number $m$
of distinct videos stored in the system, i.e. its catalog size.
In an homogeneous system (i.e. all node capacities are equal)
where a box downloads its video with no more than $c$ equal rate
connections, we give necessary conditions for achieving scalable
catalog size. In particular, we prove for that case a lower bound
$u\geq\max\set{1+\frac{1}{c},\mun}$, where $\mun\ge 1$ is the
maximum growth factor of any swarm of boxes viewing the same video during
a period of time equivalent to start-up delay (our model tolerates
swarms growing exponentially with time). On the other hand,
we prove that catalog size $\Omega(n)$ can be achieved with
a centralized scheduling algorithm when
$u\ge\max\set{1+\frac{1}{c},\mun}$, $c\ge 2$ and nodes are reliable.

Additionally, we propose a distributed connection scheduling
algorithm associated to a random video allocation scheme for
heterogeneous systems where box upload capacity is proportional to
storage capacity. It achieves catalog size
$\Omega\left({n}/{\log n}\right)$  and allows to successfully
handle a sequence of $O(n)$ adversarial events with high probability
as long as $u\ge\mun+\frac{1}{c}$.
As a special case, it can be used to solve single video distribution
with $O(1)$ reliable seed boxes, or $O(\log n)$ unreliable seed
boxes, with constant capacities.}
\RRmotcle{vidéo-à-la-demande, passage à l'échelle, pair-à-pair}
\RRkeyword{video-on-demand, scalability, peer-to-peer}
\RRprojet{GANG}
\RRtheme{\THCom} % cas de 5 themes
\URRocq % pour ceux qui sont au centre de la France
\usepackage{color}

\def\yacine#1{}%{\color{red}[Yacine: {#1}]}}
\def\loufab#1{}%{\color{red}[Fabien Mathieu: {#1}]}}
\def\fabien#1{}%{\color{red}[Fabien de Montgolfier: {#1}]}}
\def\diego#1{}%{\color{red}[Diego: {#1}]}}
\def\laurent#1{}%{\color{red}[Laurent: {#1}]}}
\long\def\jump#1\finjump{}

\newtheorem{theorem}{Theorem}
\newtheorem{lemma}{Lemma}

\newtheorem{claim}{\emph{Claim}}

\newenvironment{proof}{\noindent\textbf{Proof.}}{{}\hfill$\Box$\\}

\newcommand{\paren}[1]{\left({#1}\right)}
\newcommand{\pfrac}[2]{\left(\frac{#1}{#2}\right)}

\newcommand{\set}[1]{\left\{{#1}\right\}}

\newcommand{\eps}{\varepsilon}
\let\mun=\mu % nouveau mu

\begin{document}
\makeRR   % cas d'un rapport de recherche
%% \makeRT % cas d'un rapport technique.
%% a partir d'ici, chacun fait comme il le souhaite

\section{Introduction}

\subsection{Background}

The quest for scalability has yield a tremendous amount of work in
the field of distributed systems in the last decade. Most recently,
the peer-to-peer community has grown up on the extreme model where
small capacity entities collaborate to form a system
whose overall capacity  grows proportionally to its size. 
Historically, first peer-to-peer systems were devoted to
collaborative storage (see, e.g., \cite{saroiu,gnutella,patarin06peer}).
The academic community has proposed numerous distributed solutions
to index the contents stored in a such a system. Most prominently,
one can mention the numerous distributed hash table 
proposals (see, e.g.,
\cite{ratnasamy01scalable,rowstron01pastry,maymounkov02kademlia,stoica03chord}).
Extreme attention has then been paid to \emph{content
distribution}. There now exists efficient schemes for single file 
distribution~\cite{bittorrent}. Several proposals
were made to cooperatively distribute a stream of
data (see, e.g.,
\cite{splitstream,kostic03bullet,tran03zigzag,xu02peertopeer,pplive,prefixstream}).
The main difficulty in streaming is to obtain low delay and balanced
forwarding load.
Most recently, the problem of collaborative video-on-demand has
been addressed. It has mainly been studied under the \emph{single
  video distribution} problem: how to collaboratively download a
video file and view it at the same 
time~\cite{huangliross,chengliuwhangjin,rodriguez,allen07deploying,pplive,p2cast,BASS,janardhanschulzrinne,p2vod}. 
This somehow combines both
file sharing and streaming difficulties. On the one hand,
participants are interested by different parts of the video. On
the other hand, an important design goal resides in 
achieving a small \emph{start-up delay}, i.e. the delay between
the request for the video and the start of playback.

Most of these solutions rely on a central server for providing
the primary copy of a video to
the set of entities collaboratively viewing it.
Following the pioneering idea of Suh et al.~\cite{pushtopeer},
we propose to explore the conditions for achieving 
fully distributed scalable video-on-demand systems. 
One important goal is then to enable a large 
\emph{distributed catalog}, i.e. a large number of 
distinct primary video copies distributively stored. 
We thus consider the entities 
storing the primary copies of the videos as part of the
video-on-demand system. This model can encompass various
architectures like a centralized system with download-only clients,
a peer-assisted server as assumed in many proposed solutions,
a distributed server with download-only clients or a fully
distributed system as proposed in~\cite{pushtopeer}.
These scenarios are illustrated by Figure~\ref{fig:exemples}.
The fully distributed architecture is mainly motivated by
the existence of set-top boxes placed directly in user homes
by Internet service providers. As these boxes may combine 
both storage and networking capacities, they become an
interesting target for building a low cost distributed
video-on-demand system that would be an alternative to more
centralized systems. % LV deja mentionee plus haut (Figure~\ref{fig:exemples}).

\begin{figure}[t]
\begin{center}
\subfigure[Box description]{\includegraphics[height=3cm]{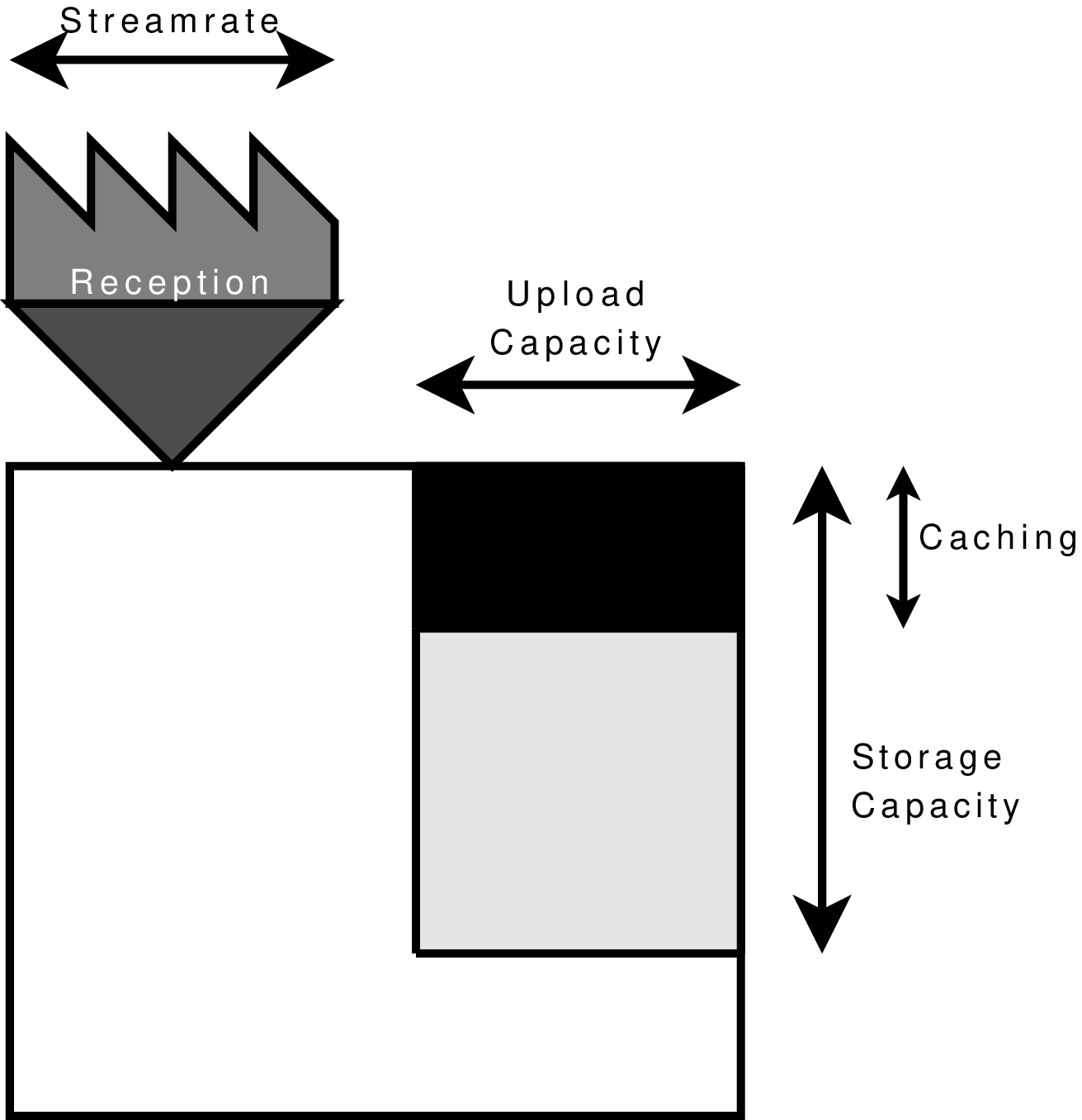} \label{fig:box}}
\subfigure[Fully centralized]{\includegraphics[height=3.5cm]{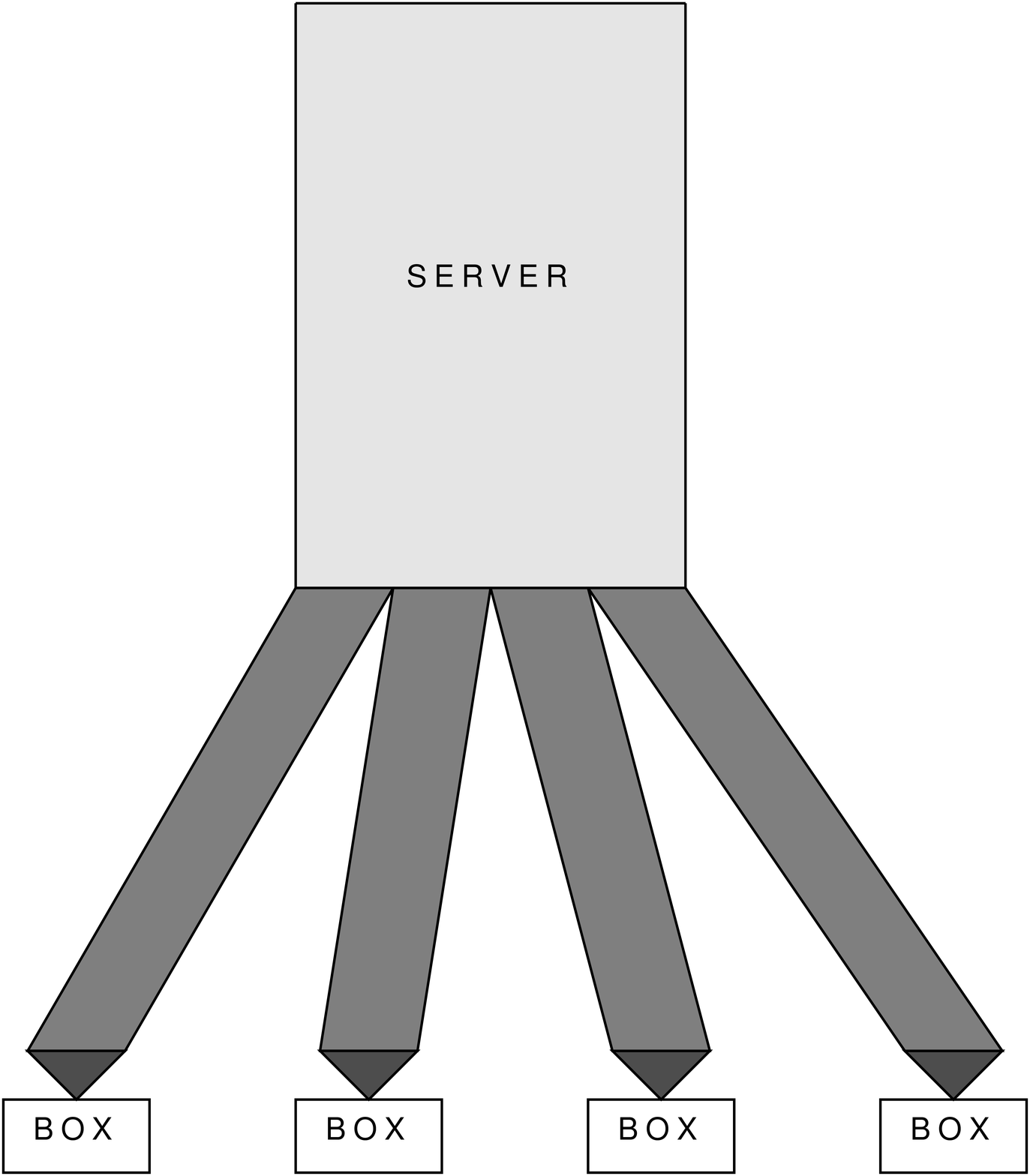} \label{fig:server}}
\subfigure[Cached servers]{\includegraphics[height=3.5cm]{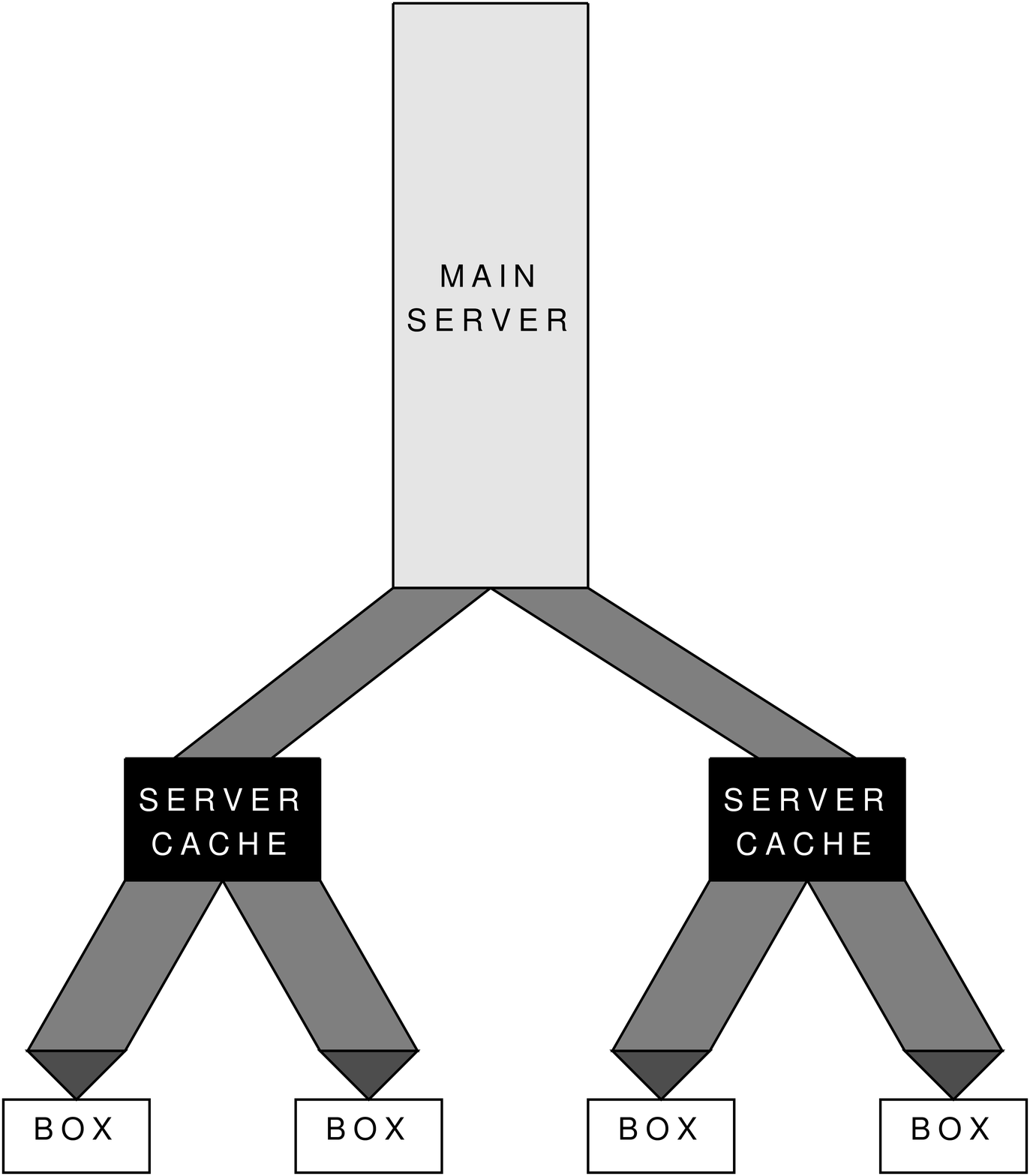} \label{fig:server_cdn}}
\subfigure[Peer-assisted]{\includegraphics[height=3.5cm]{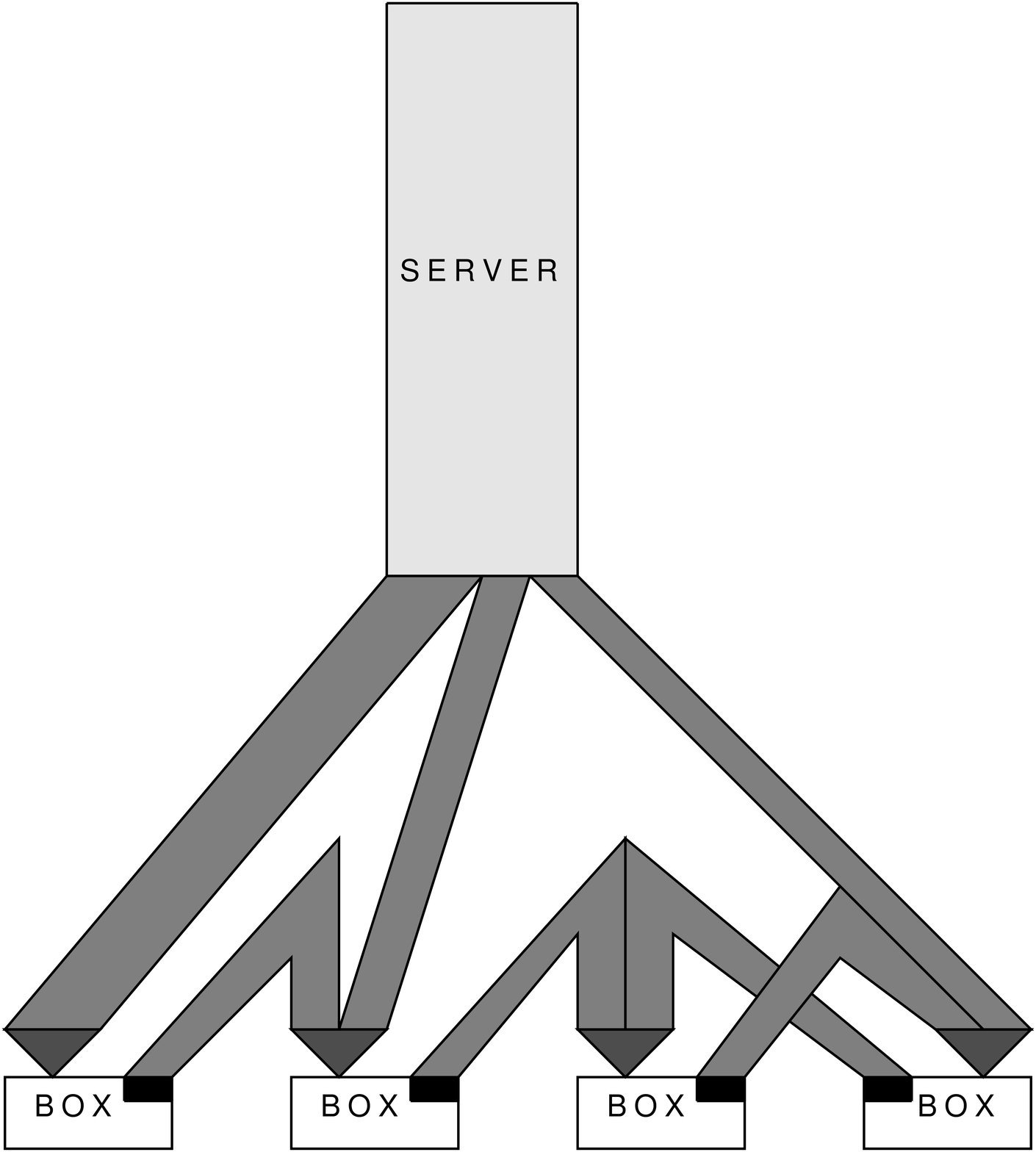} \label{fig:server_assisted}}
\subfigure[Fully distributed]{\includegraphics[height=3.5cm]{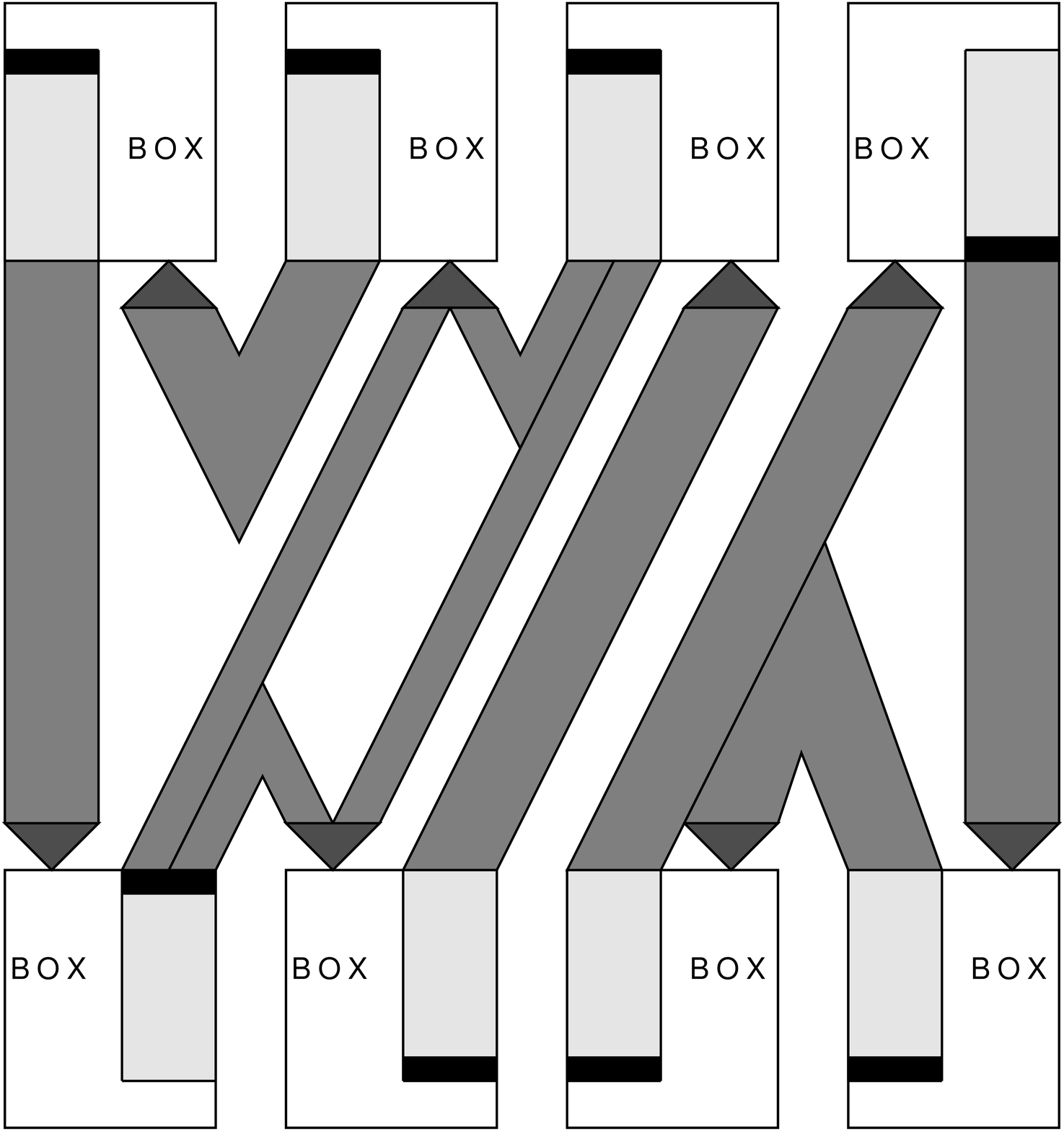} \label{fig:pushtopeer}}
\end{center}
\caption{Generic box description, and possible Video-on-Demand architectures}
\label{fig:exemples}
\end{figure}

% \begin{figure*}[t]
% \begin{center}
% \subfigure[Scenario \#1 (blocked)]{\includegraphics[width=.3\textwidth]{../figures/cexemple2} \label{fig:cex2}}
% \subfigure[Scenario \#1 (unblocked)]{\includegraphics[width=.25\textwidth]{../figures/cexemple3} \label{fig:cex3}}
% \subfigure[Scenario \#2]{\includegraphics[width=.33\textwidth]{../figures/cexemple1} \label{fig:cex1}}
% \end{center}
% \caption{Basic scenarios where VoD is compromised although storage and bandwidth are OK.}
% \label{fig:cexemple}
% \end{figure*}

\subsection{Related Work}

A significant amount of work has been done on \emph{peer-assisted} 
video-on-demand,
where there is still a server (or a server farm) which stores the
whole catalog. Annapureddy \emph{et al.}~\cite{rodriguez}
investigate the distribution (on-demand) of a single video. They
propose an algorithm that uses a combination of network
coding, segment scheduling and overlay management in order to handle
high streamrates and slow start-up delays even under
\emph{flashcrowds} scenarios. This follows an approach similar
to \cite{pplive} consisting in grouping
viewers of the same segment of the video together. 
Adaptations of the BitTorrent protocol to the single video
distribution are proposed in \cite{janardhanschulzrinne,BASS}.
 Cheng \& al. propose \cite{chengliuwhangjin} connections
to nodes at different position in the video to enable
VCR-like features (seeking, fast-forwarding, \ldots).
A thorough analysis of single video distribution under Poisson
arrival is made in~\cite{huangliross}, strategies
for pre-fetching of future content are simulated against real traces.
 % The diffusion of multiple videos has
% been proposed in \cite{huangliross,chengliuwhangjin}. 
 % In
% \cite{huangliross}, two distinct approaches are proposed. The single
% video approach assumes that peers only collaborate to the video they
% are watching. It is therefore very similar to the single video
% distribution proposed in \cite{rodriguez}. The multiple video
% approach, where a peer is allowed to redistribute a video previously
% watched, is closer to our model. However, the existence of a server
% is still mandatory.
% LV: pb je n'ai pas la meme lecture
Caching strategies are tested against real traces
in~\cite{allen07deploying}.
It is proposed in \cite{janardhanschulzrinne} to use a distributed
hash table to index videos cached by each node. However, there
is no guarantee that the videos stay in cache.
All these solutions rely on a centralized server for feeding the system with 
primary copies of videos.

To the best of our knowledge, only a few attempts have been made so
far to investigate the possibility of a server-free video-on-demand
architecture. 
 %All assume that the videos are proactively allocated
%to participating peers (like set-top boxes), and that the VoD
%service itself only relies on those peers.
% In~\cite{janardhanschulzrinne}, the main concern is the quality of
% the playback and the design of VCR-like functionality (seeking,
% fast-forwarding, \ldots) and the catalog size is not considered. 
% LV: c'est plutot chengliuwhangjin ca
Suh \emph{et al.} proposed the Push-to-Peer
scheme~\cite{pushtopeer} where the primary copies of the catalog
are pushed on set-top boxes that are used for video-on-demand.
The paper addresses the problem of fully distributing the system
(including the storage of primary copies of videos), but
scalability of the catalog is not a concern. Indeed,
a constant size catalog is achieved: each box stores a portion 
of each video. A code-based scheme is combined to a
window slicing of the videos and a pre-fetching of
every video. The paper is mainly dedicated to a complex analysis
of queuing models to show how low start-up delay and sufficiently
fast download of videos can be achieved.
The system is tailored for boxes with upload capacity lower than
playback rate. 
As we will see, this is a reason why scalable catalog cannot be
achieved in this setting.

Finally, in a preliminary work~\cite{iptps}, we begun
to analyze the conditions for catalog scalability. This work
mainly focuses on the problem of serving pairwise distinct videos
with a distributed system with homogeneous capacities
and no node failure.
Most notably, an upper bound of $n+O(1)$ is shown for catalog size
when upload is too scarce. A distributed video-on-demand is sketched
based on pairwise distinct requests and using any existing single
video distribution algorithm for handling multiply requested videos.
We extend much further this work to multiple requests, heterogeneous
case and node churn scenarios. We can now provide an upper bound of
$o(n)$ for catalog size when upload is scarce and multiple
requests are allowed. Secondly, we prove that the maximum flow
technique proposed for pairwise distinct requests can be extended
to answer any demand with possible multiplicity. This requires a
much more involved proof. Additionally, we give insight on
heterogeneous systems where nodes may have different capacities one
from another. Finally, we propose a distributed algorithm combining
both primary video copy distribution and replication of multiply
requested videos. 
Let us now give more details about the contributions of the present paper.

\subsection{Contribution}

This paper mainly proposes a model for studying the conditions
that enable scalable video-on-demand. Most importantly, we focus
on scalable catalog size and scalable communication schemes.
Our approach consists in first formulating  necessary requirements
for scalability and then try to design algorithms based on these
minimal assumptions. We call \emph{boxes} the entities forming the system.
Most notably,
we require that a box downloads a video using
a limited number $c$ of connections. This is a classical assumption
for having a scalable communication maintenance cost in an overlay
network. Note that efficient
$n$-node overlay network proposals usually try to achieve $c=O(\log
n)$. Equivalently, we assume that video data and video stream cannot
be divided into infinitely small units. With at most $c$
connections, a single connection should have rate at least
$\frac{1}{c}$ where $1$ corresponds the
normalized playback rate of the video.
Similarly, as connections have to remain steady
during long period of times with regard to start-up delay $t_S$,
a box should store portions of video data of size at
least $\frac{c}{t_S}$. This assumptions of minimal unit of data
or minimal connection rate provided by a box of the system are
particularly natural when one faces the problem of distributing video data
on several entities: one have to define some elementary chunk size
and distribute one or more of them per entity.

We first show that these discrete nature assumptions on connection
rates and chunk size give raise to an upload bandwidth
threshold. If the \emph{average upload} $u$ is no more than 1,
scalable catalog size cannot be achieved, a minimal average upload
of $1+\frac{1}{c}$ is thus required. Theorem~\ref{th:scarce}
states this as soon as $c=O(n^\eps)$ for any $\eps<\frac{1}{2}$
(e.g., $c$ is constant or bounded by a poly-logarithmic function of $n$).
Moreover, a distributed
video-on-demand system cannot achieve scalable catalog size
if the number of arrivals for a given video increases too
rapidly. We call \emph{swarm} of a video
the set of boxes playing it. If the swarm of a video can increase by
a multiplicative
factor $\mun > 1$ during a period equivalent to start-up delay
$t_S$, then it is necessary to have upload $u\ge \mun$ to replicate
sufficiently quickly the video data (see Theorem~\ref{th:swarmgrowth}).
These lower bounds on $u$ mainly rely on the assumption
that with large catalog size, some video must be replicated on a
limited number of boxes. (This assumption may be deduced from our
bound $c$ on the number of connections or may be taken for itself).

On the other hand, we give algorithms for enabling scalable
video-on-demand. We model the algorithmic part of a
video-on-demand system with two algorithms: a \emph{video allocation}
algorithm is responsible for placing video data on boxes, and
a \emph{scheduling algorithm} is responsible for managing
video requests proposed by an adversary, i.e. propose connections
for each box to download its desired video.
We build two scheduling algorithms based on random allocation of
video data.
Let us first remark that is not possible to
resist node failures if some video has its data on a limited
number of boxes: an adversary can place node failure events on
these boxes and then
request the video. We thus propose a first scheduler
under the assumption that no node fails and that
we meet the conditions $u\ge\max\set{1+\frac{1}{c},\mun}$
and $c\ge 2$. The problem of finding suitable connections
for downloading all videos reduce to a maximum flow problem
for a given set of requests and a given allocation of videos.
We thus propose a centralized scheduler running a maximum flow algorithm. 
If a centralized tracker for orchestrating connections has already
been proposed in several peer-to-peer architectures, it
is not clear whether 
this maximum flow computation could be made in a scalable way.
The benefit of 
this algorithm is thus mainly theoretical. It allows to understand the
nature of the problem. 
Theorem~\ref{th:expander} states that a random allocation enables
a catalog of size $\Omega(n)$ and allows
to manage any infinite sequence of adversarial requests with high
probability (as long as the adversary cannot propose node failures).
The problem of scalable video-on-demand can thus be solved with
optimal upload capacity in theory.
Interestingly, this scheme  allows to show that the best catalog 
size is obtained when the storage capacity of boxes is proportional
to their upload capacity.
% rien

Additionally, we propose a
randomized distributed scheduler based on priority to
\emph{playback caching}, i.e. relying on the fact that boxes playing
a video can redistribute it. Giving priority to such connections
allows to be resilient to exponential swarm growth. We show that
with the random allocation of $\Omega(n/\log n)$ videos in a system where
average storage capacity is $d=\Omega(\log n/c)$ per box,
this scheduler
can manage $O(n)$ realistic adversarial events with high probability
under the assumption that
$u\ge \mun+\frac{1}{c}$ and the adversary is not aware of
the scheduler and allocation algorithm choices 
(see Theorem~\ref{th:realistic}).
Interestingly, our use of playback caching allows to build
disjoint forwarding trees for video data in a way similar to
Splitstream~\cite{splitstream}. The main difference is that
relaying nodes buffer data before forwarding it and tree levels are
ordered according to the playing position in the video.

\smallskip

The paper is organized as follows.  Section~\ref{sec:requirements}
exposes the requirements that are needed for the catalog to be
scalable. Section~\ref{sec:strong} investigates the worst case
analysis of the problem with no failures; while
Section~\ref{sec:realistic} considers more realistic conditions.
Then Section~\ref{sec:simus} proposes to confirm the results of
previous sections by the dint of simulations. Some proofs are in
given in appendix due to space limitations.  We now introduce our
model for video-on-demand systems and the notations used throughout
this paper.

\section{Model}
\label{sec:model}

\begin{table}
\begin{tabular}{|l|p{.9\textwidth}|}
\hline \small
$n$ &\small Number of boxes for serving videos.\\ \small
$m$ &\small Number of videos stored in the system (catalog size).\\ \small
$d_i$ &\small Storage capacity of box $i$ (in number of videos).\\ \small
$d$ &\small Average storage capacity of boxes.\\ \small
$k$ &\small Number of duplicates copies of a video with random
allocation ($k\approx nd/m$)\\ \small
$u_i$ &\small Upload capacity of box $i$ (in number of full video streams).\\ \small
$u$ &\small Average upload capacity of boxes.\\ \small
%$b$ &\small Upload provisioning w.r.t. requests ($b=\frac{un}{r}$).\\ \small
$c$ &\small Maximum number of connections for downloading a video.\\ \small
%        for uploading.\\ \small
$s$ &\small Number of stripes of videos (a video can be viewed by
        downloading its $s$ stripes simultaneously).\\ \small
%$q$ &\small Number of distinct videos in a request ($q\le r$).\\ \small% ($q\le r \le n$) \\ \small
$a$ &\small Minimum ratio of active boxes in an homogeneous system.\\ \small
$t_S$ &\small start-up delay: maximum delay to start playing a video.\\ \small
$v_S$ &\small Maximum number of arrivals during $t_S$ for a video
  not being played.\\ \small
% LV : vraiment besoin ici ?
% $w(v)$  &\small swarm size: number of boxes viewing video $v$.\\ \small
$\mun$ &\small Bound on swarm growth: if a swarm has size $p$ at time
  $t$, it has size less than $\mun p$ at time $t+t_S$.\\
%% $\ell,\alpha$ & Bound on arrival rate: at most $\ell\alpha^{(-t/t_S)}$
%%  during a period of time $t$.\\ \small
%% plus besoin d'exponential global -- c'etait bien la peine que 
% je change tous mes alpha ren mu -- FdM
% INUTILE : $\alpha$ & MINIMAL ratio of  boxes NOT in startup phase.\\
\hline 
\end{tabular}
\caption{Key parameters}
\label{tab:parameters}
\end{table}

We first introduce the key concepts of video-on-demand systems
and discuss the associated parameters. We first describe
the nodes (often called boxes) of the system, then detail how they may
connect to each other to exchange data. We then explain how we decompose
the algorithmic part of the system and describe adversary models for testing
our algorithms.

\paragraph{Video system.}
We consider a set of $n$ boxes used to serve videos among themselves.
Box $i$ has storage capacity
of $d_i$ videos and upload capacity equivalent to $u_i$ video streams. For
instance if $u_i=1$, box $i$ can upload exactly one stream (we
suppose all videos are encoded at the same bitrate, normalized at $1$). 
Such a system will be called an $(n,u,d)$-\emph{video system}
where $u=\frac{1}{n}\sum_{i=1}^n u_i$ is the average upload capacity
and $d=\frac{1}{n}\sum_{i=1}^n d_i$ is the average storage capacity.
A system is \emph{homogeneous} when $u_i=u$ and $d_i=d$ for all $i$.
Otherwise, we say it is \emph{heterogeneous}. The special case when storage capacities
are proportional to upload capacities (i.e. $d_i=\frac{d}{u}u_i$
for all $i$) is called
\emph{proportionally heterogeneous}.

The box activity is defined as a \emph{state}. Box $i$ is
\emph{active} when it can achieve a stable upload capacity no less than $u_i$
or \emph{inactive} otherwise (e.g. when it is under failure or
turned off by user).  We suppose that the ratio $a$ of active boxes
remains roughly constant.  We assume that the nodes with higher capacity are not more prone to failure than the other nodes, so the average upload capacity of active boxes remains larger than $u$.
An active box may be \emph{playing} when
it downloads a video or \emph{idle} otherwise.
The set of boxes playing the same video $v$ is called \emph{swarm}.  Node churn occurs as
sequence of \emph{events} consisting in changing the state of a box.
\emph{Swarm churn} designates the events concerning a given swarm.
We will see in Section~\ref{sec:arrivals} that scalability cannot be
achieved when a swarm grows too rapidly. We thus assume a bounded
\emph{growth} factor $\mun$: during a period of time $t_S$ ($t_S$ is defined below), the
size of a swarm is multiplied by a factor  $\mun$ at most.
More
precisely, and to remove any quantification issues, we assume that the number of events for a given swarm and
a given period of time $t$ is at most $\mun^{t/t_S}$.  (For convenience, we
aggregate the various types of swarm churn within the same bound).

\paragraph{Connections.}
We assume that finding, establishing and setting up
a small buffer for starting video playback takes time. 
We call \emph{start-up delay} the maximal duration $t_S$  for a box
to connect to other boxes and begin playback.
We consider that the number $c_n$ of
connections for downloading a video is bounded by some constant $c$.
The reason is that with constant swarm churn rate, 
a box will have to change $\Omega(c_n)$ 
connections per unit of time. 
As changing a connection has some latency $\Omega(t_S)$,
this number should remain bounded or grow very slowly with
$n$. In connection with this assumption, we suppose that the
data of video cannot be split in infinitely small pieces.
We thus consider that a connection has minimal rate $\frac{1}{c}$
(this is obviously the case when connections rates are equally
balanced, and it can be modeled by aggregating unitary connections
otherwise). Therefore the minimal piece of video data stored on a box is
$\Omega\paren{\frac{1}{c}}$
(a trivial lower bound of $\frac{t_S}{c}$ follows from previous
assumptions).

A peer-to-peer video system without any external video sourcing
relies on the possibility to replicate a video as it becomes 
more popular
and the number of requests for the video increases.
The most straightforward way to do this is to cache in each box the
video it is currently playing, which is natural if we want to provide some VCR functionalities. We call \emph{Playback caching} 
this facility: boxes of a swarm can serve as a relay for for the boxes
viewing a former part of the video. Note, that in order to bring some flexibility in the swarm, the video can be split into time windows, thus allowing to avoid linear viewing. Time windowing also allows to reduce the problem to the case where all videos have
approximately the same duration.

\paragraph{Video data manipulations.}
We consider that all videos have
same playback rate, same size and same duration
(all three equal to $1$ as they are taken as reference for
expressing quantities). 
% We could manage various playback rates,
% sizes or durations
% similarly to how we manage heterogeneity of box capacities. 
% For the sake of clarity, we do not go in such details.  
%
To enable multi-source upload of a video, each video may be divided in
$s$ equal size \emph{stripes} using some balanced encoding scheme.
The video can then be viewed by downloading simultaneously the
$s$ stripes at rate $1/s$.
A very simple way of achieving stripping consists in
splitting the video file in a sequence of small packets. 
Stripe $i$ is then made of the packets with number equal to
$i$ modulo $s$.
Note that our connection number limitation imposes $s\le c$.
 % -- Cutable:
There are two main reasons for using stripes:
it allows to build internal-node-disjoint trees as discussed in
Section~\ref{sec:realistic} and it let a box upload sub-streams of
rate $\frac{1}{s}$ to fully use its upload capacity.
Stripes may also enable redundancy through correcting codes
at the cost of some upload overhead: downloading 
$(1-\eps)s$ stripes is then sufficient to decode the full video
stream (e.g. using LT-codes~\cite{ltcodes} or rateless
encoding~\cite{rateless}).
For the sake of simplicity, we assume that
$s$ can be large enough to consider all $u_i s$ and
$d_i s$ as integrals.
As mentioned previously, a video can be distributed among
several boxes by splitting it according to time
windows. However, considering all the time windows of all videos
being played at given time, we are back to the same problem
fundamentally. For that reason, we do not develop time windowing.
% --

\paragraph{Video scheme.}
A \emph{video allocation} algorithm is responsible for placing
primary copies of each video in the system respecting
storage capacity constraints of boxes. The most simple
scheme consists in storing them
statically: video data may be replicated but primary
copies of videos are static.
Video allocation only changes when new boxes are added to the system
or when the catalog is updated.
For instance, re-allocating primary copies under node churn would
not be practical when live connections consume most of the upload
capacity of the system.
We assume that the catalog renewal is made at a much larger
time scale. Its size and storage allocation are thus considered
fixed during a period of several
playback times. 
We may assume that the catalog remains the same
during such periods.
Of course as the system evolves over a long period of time,
some videos are added or removed.
The \emph{catalog size} is the number of distinct videos allocated.

When a box state changes, a \emph{scheduling algorithm} decides how to update 
the connections of playing boxes so that their video is downloaded
at rate greater than 1 and all box upload capacities are respected.
For our theoretical bounds (Section~\ref{sec:strong}), we use a centralized scheduler
that has full knowledge of the system.
For practical algorithms ( Section~\ref{sec:realistic}), we consider
distributed scheduling algorithms:
each time a box changes it state, it runs the scheduling algorithm
on its own. 
The scheduling algorithm succeeds if it can establish connections
to download the full video stream in time less than $t_S$.

We call \emph{video scheme} a combination of an allocation scheme
and a scheduling algorithm.
We say that a video scheme \emph{achieves} catalog size $m$ if the
allocation scheme can store $m$ videos in the system so that the
scheduling algorithm succeeds in handling all requests of an
adversary.
The \emph{adversary} knows the list of videos
in the catalog and proposes any
sequence of node state changes that respects our model assumptions.
In its weaker form, it is not aware of  the decisions made by 
the allocation and scheduling algorithms. This is a realistic assumption 
as there is no reason for user requests to be correlated
to something the users of the system are not aware of.
 %(Of course it knows which box is playing what since it decides this). 
Worst case analysis is obtained with the \emph{strong adversary}
which is the most powerful adversary possible.
It is additionally aware of the choices made
by the allocation and scheduling algorithms. In particular, it
knows which boxes contain replicas of a given video. If not specified,
the adversary is not strong.

\section{Necessary Conditions for Catalog Scalability}
\label{sec:requirements}

Let us first give some trivial requirements. 
The total upload is at most $un$ and, 
as all active boxes may be playing, the total download 
capacity needed may be $n$
so we trivially 
deduce the lower bound $u\ge 1$.
As the total storage
space of active boxes can be as low as $adn$ (assuming that
average storage capacity of active boxes remains $d$), we have
$m\le adn$.

Let us first remark that if we release the constraint on bounded connectivity,
then ideal storage of $adn$ videos can theoretically be 
achieved in any proportionally heterogeneous
$(n,1,d)$-video system when $c=n$. 
As stated in the homogeneous
case~\cite{pushtopeer}, \emph{full stripping} can achieve this.
It consists in splitting each video in $n$ stripes, one per box.
Viewing a video then requires to connect to all other boxes.
This result can easily be generalized to the proportionally heterogeneous
case with node failures using correcting codes. 
Such scheme are unpractical  for large $n$ but give a
theoretical solution.

On the other hand, we show that some upload provisioning is
necessary in our more realistic model.
The main hypothesis implying these results is that some video
is replicated on $O\left(\frac{n}{m}\right)$ boxes at most, i.e.
$o(n)$ if catalog scales.
First note that as soon as a video spans at most $o(n)$ boxes, the
system cannot tolerate $n$ strong adversarial events.  Indeed,
the strong adversary can propose failure events on all boxes possessing
a given video and then propose a request with the video.

\subsection{Maximal Swarm Churn Rate}
\label{sec:arrivals}

We now state that arrival rate in a given swarm must be lower than
average upload.
This is our first non trivial lower bound on average upload.

\begin{theorem}
\label{th:swarmgrowth}
Any homogeneous
$(n,u,d)$-video system achieving catalog size $m$ 
and resilient to swarm growth $\mun$ 
satisfies $u\ge\max\set{2, \mun -O\pfrac{1}{m}}$
\end{theorem}

For small start-up delay, a realistic value of $\mun$
would  certainly be less than $2$.
Scalable catalog size is then achievable for $u\ge \mun$ only.

\begin{proof}%[of Theorem~\ref{th:swarmgrowth}]
We consider a scenario where  boxes are viewing  different
videos, and all of them switch to the same video forming a swarm
with growth factor $\mun$.
The swarm of the video has thus size $v_S$ at time 0,
$v_S \mun$ at time $t_S$, $v_S \mun^2$ at time $2t_S$,
and more generally size $v_S \mun^i$ at time $it_S$.
We choose a video that is replicated at most
$k=O\paren{\frac{n}{m}}$ times in the system.
If this data is possessed by sufficiently many boxes,
it can be replicated $k$ times initially.
%by at most $k=O\paren{\frac{n}{m}}$ boxes originally.
Consider the number of times $x_i$ the data of the video 
is replicated outside the swarm at time $it_S$.
Suppose that all boxes possessing the video
either serve new arrivals or 
pro-actively replicate it with their remaining bandwidth.
We then have $v_S \mun^{i+1}+x_{i+1}\le v_S u\mun^i + (u-1)x_i$
as the video data must be received by
all boxes in the swarm and boxes outside the swarm that replicate it.
Suppose $u\le \mun$  (otherwise the proof is already over). %As $v_S<k$
We get $x_{i+1}\le (u-1)x_i$, and thus $x_{i}\le (u-1)^ik$
as $x_0\le k$. The former inequality thus gives
$u+\frac{k}{v_S}\pfrac{u-1}{\mun}^i\ge\mun$. 
%$i$ can not grow arbitrarily large: as $x_i\le n$ then $i\le \log_{\mun} n$.
If $u<2$, we obtain
for $i=\log_{\mun} n$ that
$u\ge\mun - \frac{k}{v_S n}=\mun - O\pfrac{1}{m}$.
\end{proof}

Additionally, we can prove that a strict upload of 1 is not
sufficient even under low pace arrivals. 

\subsection{Upload Capacity versus Catalog Size}
\label{sec:scarce}

We thus assume in this section that $c=O(n^\eps)$ for some
$\eps>0$.
This is for example the case 
when $c$ is a poly-log of $n$ as often assumed in overlay 
networks~\cite{ratnasamy01scalable,rowstron01pastry,stoica03chord}. 
(The rest of the paper assumes  a constant $c$).
With this bound, we can establish the following trade-off
between average upload capacity and achievable catalog size.

% maintained by a box is bounded by $c=o(n)$ and that these
% connections should remain steady during a period of time $t$
% significantly long. (A system with too many connections changing too
% frequently would have less resilience to failure and high overhead.)

\begin{theorem}
\label{th:scarce}
For any $\eps>0$,
an homogeneous $(n,u,d)$-video system with $u\le 1$ and $c=O(n^\eps)$
%\laurent{TODO: seulement $u\le \frac{1}{c-1}$}
that can
play any demand of $n$ videos in the no failure strong adversary model
%\laurent{TODO: randomized adversary (ac modele aleat de req)} 
has catalog size $m=O\paren{n^{1/2+\eps}}$.
\end{theorem}

The above result states that a video system with scarce capacity poorly
scales with $n$. As it is valid in the no failure strong adversary
model, it remains valid in the strong adversarial model.
With our discrete vision of connections, it implies that a minimal
upload $u\ge 1+\frac{1}{c}$ is necessary for scalability.

\begin{proof}
Suppose there exists $\eps>0$ with $\eps<\frac{1}{2}$ such that $c< n^\eps$.
As discussed in Section~\ref{sec:model}, we use our assumption
that a box stores 
no less than $\frac{t_S}{c}$ data of a given video.
%corresponding to a time window
%of duration $t$ of a video in a box for some constant $a>0$.
%\laurent{TODO: passer la preuve a $c<n^x$ et virer le $1/t$ dans la borne.}

Suppose by contradiction that there 
exists a video system with catalog size
$m>\frac{2d}{t_S}n^{1/2+\eps}>\frac{2dc}{t_S}\sqrt{n}$.
As the overall storage capacity is $dn$, there 
exists some video $v$ whose data is replicated at
most $\frac{dn}{m}\le \frac{t_S}{2c}\sqrt{n}$ times. As useful portion of
data of $v$ have size
at least $\frac{t_S}{c}$, the set $E$ of boxes storing data of $v$
has size at most $\frac{1}{2}\sqrt{n}$. Let
$F=\overline{E}$ be its complementary. Set
$p=|E|\le\frac{1}{2}\sqrt{n}$ and $q=|{F}|$.

Now consider the possible request
sequence where all boxes $b_1,\ldots,b_q$ 
of  ${F}$ successively begin to play $v$
while boxes of $E$ play videos 
not stored at all among boxes in $E\cup\set{b_q}$.
Box $b_i$ can  download $v$ from
$E_i=E\cup\{b_1,\ldots,b_{i-1}\}$. 
Boxes of $E$ can only download from $F'=F\setminus\set{b_q}$.

Suppose that data of $v$ flows from $E$ to $F'$ at rate $p'$
and from $E$ to $b_q$ at rate $p''$.
We have $p'+p''\le p$ since
the overall upload capacity of $E$ is $p$.
Data of $v$ flows internally to ${F'}$ at rate at least
$q-1-p'$. The remaining upload capacity to serve $E$ is thus
$p'-(1-p'')\le p-1$ as $E$ must additionally serve $b_q$ at rate $1-p''$. 
This implies that the number of videos not stored at
all on $E\cup\set{b_q}$ is at most $p-1$. (Otherwise, we have a request that
cannot be satisfied.)

As a box contains
data of $\frac{dc}{t_S}$ distinct videos  at most.
We thus deduce $m\le \frac{dc}{t_S} (p+1) + p-1
\le \frac{dc}{t_S}\sqrt{n}<\frac{d}{t_S}n^{1/2+\eps}$.
This is a contradiction and we deduce $m=O\paren{n^{1/2+\eps}}$.
\end{proof}

We deduce from the previous results that
$u\ge\max\set{1+\frac{1}{c},\mu}$ is a minimal requirement for
scalability. We now show that it is indeed sufficient.

\section{Strong Adversary Video Scheme}
\label{sec:strong}

We now propose  a video scheme
achieving catalog size $\Omega(n)$ in the no failure
strong adversary model for any video system
with average upload $u\ge \max\set{1+\frac{1}{c},\mun}$.
It is based on random allocation of video stripes using
$s=c$ stripes per video and uses a maximum flow scheduler.

% by the following theorem.

% \begin{theorem}
% \label{th:expander}\hspace{.5cm}
% Any proportionally heterogeneous $(n,u,d)$-video system
% with $u\ge 1+\frac{1}{c}$, can achieve
% catalog size $\Omega(dn/\log_{u}d)$ in the no failure strong
% adversarial model.
% \end{theorem}

% This result is quite strong as the multiset of requests may be chosen 
% by an adversary knowing the allocation graph.
% % (This choice includes
% % which boxes have reduced upload capacity due to single video schemes).
% It states that there exists a video allocation that an adversary
% cannot defeat.

% The benefit of this theorem is mainly theoretical as it requires 
% a centralized scheduler and a maximum flow computation each time
% a request arrives. We first introduce a scheme achieving
% these bounds for the homogeneous case. When then generalize to
% the heterogeneous case.

\subsection{Random Allocation}
\label{sec:randbound}

Random allocation consists in storing $k$ copies of each stripe
by choosing $k$ boxes uniformly at random. This approach was
proposed by Boufkhad \& al~\cite{iptps} using a purely random
graph with independent choices. This has the disadvantage
to unbalance the quantity of data stored in each box. 
We thus prefer to consider 
a regular bipartite graph  where all storage space is used on all
boxes.
%\laurent{intro:}
We could obtain the same bounds for the purely random graph.
Analysis is slightly more complicated in our case.

For the sake of simplicity, we assume $k=dn/m$ is an integer.
A regular random allocation consists in copying each stripe in $k$
boxes such that each box contains exactly $ds$ stripe copies.
We model this through a random permutation $\pi$ of the $kms$ stripe
copies into the $dns$ storage slots of the $n$ boxes together:
copy $i$ is stored in slot $\pi(i)$ (the $d_1s$ first slots fall into
the first box, the $d_2s$ next slots into the second box, and so on).
The best catalog size is obtained for the smallest possible value of
$k$. 

We call \emph{random allocation scheme} the video allocation
algorithm consisting in selecting uniformly at random
a permutation $\pi$ and in allocating videos according to $\pi$.

\subsection{Maximum Flow Scheduler}

We propose a connection scheduler relying on playback caching.
Each time a node state changes, a centralized tracker
considers the \emph{multiset of stripe requests}, 
i.e. the union of all the video
stripes being played (some stripes may be played multiple
times) and tries to match stripe requests against boxes so that  box
$i$
has degree at most $u_i s$.
We can model 
this problem as a flow computation in
the following bipartite graph between stripe requests and the boxes 
storing these stripes. An arc of capacity $1$
links every stripe request to all boxes where it is stored
(either through the static allocation scheme or through playback
caching). The scheduling algorithm consists in
running a maximal flow algorithm
to find a flow from stripe requests to boxes with the following
constraints: each request
has an outgoing flow of $1$ and such that box $i$ has incoming flow
of $u_is$ at most. 

We prove
that a random regular graph using $s\le c$ stripes
with $u\ge \max\set{1+\frac{1}{s},\mun}$
has the following property with high probability: 
for any multiset of  $n$ requests at most, 
a flow with the desired constraints exists. The proof consists
in proving that a random regular allocation graphs
has some expander property with high probability. 
A min-cut max-flow theorem allows to conclude and state the
following theorem.

\begin{theorem}
\label{th:expander}\hspace{.5cm}
Consider a proportionally heterogeneous $(n,u,d)$-video system
with $u\ge\max\set{1+\frac{1}{c},\mun}$ and $c\ge 2$.
Random regular allocation combined with the maximum flow scheduler
allows to achieve 
catalog size $\Omega(dn/\log_{u}d)$ and to manage successfully
any infinite sequence of strong adversarial events excepting node
failures with high probability.
\end{theorem}

% A logarithmic  number of trials would allow to find such an ideal
% allocation graph surely. This allows to prove Theorem~\ref{th:expander}
% in the homogeneous case (as $s\le c$ we have $1+\frac{1}{s}\le 
% 1+\frac{1}{c}$).

The proof generalizes in a non trivial manner the proof
of~\cite{iptps} that assumes a purely random graph allocation, pairwise
distinct requests and homogeneous capacities.
Due to space limitations, the proof is given in Appendix~A.

\subsection{Heterogeneous Capacities}

%We show now
%in Section~\ref{proof:expander} 
%\fabien{heu, y'a rien la dessus en section 8... est-ce bien normal ?}
As discussed in Appendix~A,  in the case of heterogeneous
capacities, the proof requires
 %be generalized to proportional heterogeneous capacities when all
 %uploads are greater or equal to $\mun$ (follwing exponential growth
 %assumption) when 
the following balance condition.  For
all set $E$ of boxes with overall upload capacity $U_E=\sum_{b\in E}
u_b$ and overall download capacity $D_E=\sum_{b\in E}d_b$ we have for
some $u'\ge \mun+\frac{1}{s}$:
 % \fabien{j'ai corrige un $1/c$ en $1/s$ ici c'est OK ? En d'autres
%   termes les uploads sont des multiples de $1/s$ ou $1/c$ finalement ?
% SInon changer a plusieurs endroit dans cette section des $1/s$ en $1/c$ }
$$
\frac{U_E}{D_E}\ge  \frac{u'}{d}
$$
(The number of copies per stripe in the allocation graph is 
then $k=O(\log_{u'} d)$).

Note that $u' = u$ in the proportionally heterogeneous case
and that $u'\le u$ in general. Having storage capacity proportional
to upload capacity is thus the best situation to optimally benefit
from the box capacities.

In the general heterogeneous case, a possible random allocation
scheme consists in using only storage $d_b'=d\frac{u_b}{u''}$
 %\fabien{heu c'est pas e $d_b'=d\frac{\alpha u'}{u_b}$ ?}  
for each box $b$ for some $u''\ge u$ achieving best storage capacity.
If box upload capacities are within a constant ratio, this will
achieve a catalog size within a constant ratio of the balanced
scheme. 

\subsection{Poor Upload Capacity Boxes}

Special care has to be taken for
an heterogeneous $(n,u,d)$-video system where some boxes have upload
capacity smaller than $\mun$. We say that such boxes are \emph{poor}. 
The above connection scheduler may be
defeated by downloading the same video on a large set $E$ of such
poor boxes, as it may not support exponential growth. 
This comes from the fact that the storage space for the
video coming from playback caching may get larger than $U_E$.
The above condition on the balance between storage and upload 
is then violated by playback caching storage. 
% With constant
% allocation storage and increasing playback caching storage
% there may appear some set $E$ violating the balance condition.

The general heterogeneous case is reduced to the case where
uploads capacities are all greater or equal to $\mun$ 
thanks to the following lemma. (This is the last step of
the proof of Theorem~\ref{th:expander}).
Due to space limitations, the proof is given in Appendix~A.

\begin{lemma}
\label{lem:reduction}
Consider an $(n,u,d)$-video system $A$ with $n_P$ boxes of upload less
than $\mun$ having overall upload capacity $U_P$ and a video
allocation scheme with $s$ stripes satisfying $u\ge \mun$.
There exists an
$(n,u,d+\frac{n_P-U_P/\mun}{n})$-video system $B$ with same video
allocation and, for each box $b$, upload capacity $u_b'$ satisfying
$\mun\le u_b'\le u_b$, and same average upload $u$, that can emulate
any scheme of $A$ in the no node failure strong adversary model.
% the $n_R$ boxes
% having upload $\ge 1+\frac{1}{s}$ satisfies $U_R-n_R\ge n_L - U_L$.
\end{lemma}

The idea behind this reduction is to statically reserve some
upload bandwidth of rich boxes to poor boxes. The average upload
of both systems is thus the same. When
a poor box $b$ with upload $u_b<\mun$ downloads a video, it
directly downloads $u_bs/\mun$ stripes as in the scheduling of $A$
and downloads the others through relaying by the rich boxes 
it is associated to. The rich boxes insert also the stripes they
forward in their playback cache. This explains why more storage 
capacity is required. Proof is given in Appendix~A.
% The details of the proof and the necessary
%conditions on the uploads of poor boxes versus boxes with high
%capacity is given in Section~\ref{proof:expander}.\fabien{heu, ils me
%  semble que les details viennent maintenant non ?? OU y'avait plus
%  qui a saut ?}

\section{Distributed Video Scheme}
\label{sec:realistic}

 %the same random allocation scheme than in Section~\ref{sec:randbound}. 

\subsection{Purely Random Allocation}
The video are stored in the boxes according to a purely
random allocation scheme:
each stripe of a video is replicated $k$ times. 
$s$ still denotes the number of stripes per video used.
Each replica is stored in a box chosen independently at random.
Box $i$ is chosen with probability $\frac{d_i}{dn}$.
It is possible to add a video in the system as long as the $k$
chosen boxes have sufficient remaining storage capacity.
Such an allocation scheme is qualified as \emph{purely random}.
% It is a classical balls and bins problem to show that
% $\Omega\pfrac{dn}{\log n}$ videos can be allocated with high probability.
% \laurent{Yacine ? justificiation un peu plus etaille ?}

% \begin{claim}
% catalog size $\Omega(n/\log n)$.
% \end{claim}

\subsection{Playback Cache First Scheduler}
We now propose a randomized distributed scheduling algorithm.  
The main idea of our scheduler is to give priority to
playback cache over allocated videos to allow swarm growth $\mun$. 
Only one upload connection
is reserved for video allocation uploading. An average upload
$u\ge \mun+\frac{1}{s}$ will thus be required.
The scheduling algorithm
is split in two parts: stripe searching and connection granting.

\emph{Stripe searching} is the algorithm run by a box for finding another 
box possessing a given stripe. This algorithm relies on 
a distributed hash table (or any distributed indexing algorithm)
to obtain information about a given stripe. This index allows a box
to learn the complete list of  boxes possessing the stripe through
the video allocation algorithm and a partial list of boxes
in the video swarm (i.e. boxes playing
the video of the stripe). Stripe searching consists in probing
the boxes in these lists until a box accepts a connection for
sending the stripe. A connection request includes
the stripe requested and the \emph{stripe position} in the
stripe file (i.e. an offset position indicating the
next octet of video data to be received). A box is eligible 
for a connection if it has sufficiently many video stripe data ahead
that position and if it has sufficiently many upload. This is
decided by the connection granting algorithm of the box receiving the request.
To give priority to playback-cache forwarding,
boxes of the allocation scheme are probed only when the swarm size
is less than $v_S$ or when a stripe is downloaded from a video
allocation copy less than $v_S$ times.
%\laurent{Et formul\'e comme \,ca, \,ca te va ?}
% \fabien{si y'a des deconnections apres a tient pas. Faut stocker le
%  nombre de visionneurs primaires qq part. Ou bien taper sur le
%  stckage primaire seulement apres que le cache ait chou (ce qui
%  n'est pas tres bon pour le startup delay mais on s'en fout !)} 
To balance upload,
several boxes are first probed at the same time, and an accepting 
box with least number of upload connections for the requested
video is selected.

\emph{Connection granting} is the 
algorithm run by a box that is probed for a connection request.
Suppose box $x$ receives a connection request from box $y$ for a stripe of
video $v$. The connection granting algorithm consists in the
following steps.
\begin{enumerate}
\item If box $x$ is not viewing $v$ and is already uploading the
  stripe, it refuses.
\item If box $x$ has sufficient upload capacity, it accepts.
\item Otherwise, if box $x$ is not playing $v$, it refuses.
\item Otherwise, if the stripe position of $x$ for that stripe is
  not sufficiently ahead the
  requested stripe position, it refuses.
\item Otherwise, if two or more upload connections of box $x$ concern
  a stripe of a video different from $v$, $x$ selects one of them at random,
  closes it and accepts box $y$.
\item Otherwise, if box $x$ is uploading the same stripe to some box
  $z$ and and the requested stripe position of $y$ is sufficiently
  ahead the stripe position of $z$, it closes the connection to
  $z$ and accepts.
\item Otherwise, it refuses.
\end{enumerate}

Note that Steps~4, 5 and~6 can be executed only if box $x$ plays $v$.
Step~6 can be executed only if it uploads $us-1$ stripes of
video $v$. (One connection is always reserved to serve allocated stripes).
A simple optimization in Steps~6 and~7, consists
in \emph{connection flipping}. In Step~6, box $x$ can send
to box $z$ the address of box $y$ for re-connecting
as the stripe position of $y$ is
sufficiently ahead the stripe position of $z$ in that case.
%\fabien{les graphes de De Bruijn ne sont plus tres loin !!}
Box $z$ can then probe box $y$ with the same algorithm.
In Step~7, box $y$ can be redirected to any box $x'$ downloading
$v$ from $x$ and having stripe position sufficiently ahead
the stripe position of $y$. Box $y$ can then probe box $x'$ with
the same algorithm. This way, a box can find its right position
according to stripe position in a downloading tree path of its swarm.
Similarly, in Step~4, box $y$ can make a connection flipping
with the box from which $x$ is downloading, and go up the
downloading chain until it finds its right position.

Note that this algorithm works in similar manner as
Splitstream~\cite{splitstream} builds parallel multicast trees
for each stripe.  The main difference is that each internal node of
a tree receives fresh data in a buffer and forwards data which is at
least $t_S$ old. That way, a performance blip within one node will
not percolate to all nodes behind it in the sub-tree.  Moreover, this
ensures that a node has sufficient time to recover from a parent
failure.  In addition, trees are ordered according to stripe
position: boxes with foremost playing position in the video get
closer to the root whereas newcomers in the swarm tend to be in
lower tree levels. Another interesting point is that nodes
downloading from a box with spare number of connections benefit from
this free upload capacity and download at a rate faster than needed,
allowing to fill their buffer.

\subsection{Correctness}

We cannot prove the resilience of our video scheme against any
sequence of adversarial events. The following technical assumption
is necessary for our proof and appears as a realistic hypothesis.
We assume that 
a given stripe is searched at most $O(\log r)$ times on boxes
storing it through the video allocation scheme. This requirement
is met when the sequence of adversarial events respect the
two following conditions. First, 
a constant number of swarms are started on a
given video (a realistic assumption if we consider a period of
few playback durations). (There is no restriction on swarm
size). Second, node failures are randomly chosen and a given
box is chosen with probability $p_f<1/v_S$.
A sequence of $r$ requests is said to be \emph{stress-less} if
it satisfies these conditions.

\begin{theorem}\label{th:realistic}
Consider a proportionally heterogeneous
$(n,u,d)$-video system with $u\ge \mun +\frac{1}{c}$ and
$\frac{dc}{u}=\Omega(\log n)$.
%operating during an interval time $T$ of several playback durations
For any bound $r=O(n)$,
it is possible to allocate $\Omega({n}/{\log n})$ videos
and successfully manage $r$ adversarial stress-less events
with high probability.
%starting from an arbitrary allocation of resources respecting the model.
\end{theorem}

To prove this theorem, we analyze a simpler \emph{unitary} video system
which can be emulated by any proportionally heterogeneous system
with same overall capacities. Again, we choose to use $s=c$ stripes
per video and assume $u\ge \mun+\frac{1}{s}$.
% \fabien{la y'avait encore le $1/s$ que j'ai vir. de meme que le
%   $+1/c$ dans l'enoncee du theoreme.  Ca correspond plus au lemme 1
%   sinon !}
% LV: GRRR c'est pas pour le lemme 1!
We view each box $i$ as the union of $u_i s$ unitary  boxes
with upload capacity $1/s$  (one  stripe) and storage capacity
$\frac{d_i}{u_i}=\frac{d}{u}$. This reduction is indeed penalizing.
Consider two unitary boxes that are part of the same real box.
In the model, stripes
stored on one unitary box can not be uploaded by the other whereas
the real box could use two uploads slots for any combination of
two stripes of any of the unitary boxes.
For some parameter $k$ made explicit later on, a random 
allocation of $k$ replicas per stripe is made according
to the purely random allocation scheme described previously.
This is equivalent to suppose that each replica is stored
in a unitary box chosen uniformly at random since the
system is proportionally heterogeneous.
As each unitary box has a storage capacity of $\frac{ds}{u}$ stripes,
Chernoff's upper bound allows to conclude that purely random
allocation of $\Omega(dsn/u)$ stripe replicas is possible
with high probability when $\frac{ds}{u}=\Omega(\log n)$. 
As we will use $k=O(\log n)$,
this achieves the required catalog size.

%\fabien{ca me semble plus trop correspondre au purely
%  random allocation du debut de la section non ?}
% LV:exact

Second, we simplify the scheduler to an algorithm 
where two  schedulers compete. One is allocating 
\emph{cache stripe} requests within a swarm (i.e. the stripe
will be downloaded from a playback caching copy), 
the other is allocating \emph{seed stripe} requests
from the video allocation pool (i.e. the stripe will be downloaded
from a unitary box possessing it through the random allocation scheme). 
We consider 
that both scheduler operate independently. This is a penalty with
regard to practical scheduling, where simple heuristics may reduce
considerably the number of conflicts, but it simplifies the
stochastic analysis of the system. The cache scheduler allocates
swarm stripes and has priority: it operates at real box level 
according to the above algorithm. From the unitary box point of
view of the seed scheduler, the cache scheduler disables some
unitary boxes. If the unitary box was
uploading some allocated stripe, it is canceled and a seed stripe
search is triggered. This is where the reservation of one seed
stripe per real box is useful in our analysis. A stripe upload connection is
canceled when the real box has at least two of them.
As the cache scheduler cancels one box at random uniformly, a given
seed stripe is searched at most $O(\log n)$ times with high probability.
The seed scheduler scans the list of
unitary boxes possessing the stripe until a free one is found.

Note that a video request in the real system triggers at most $s$
cache stripe requests and/or $s$ seed stripe requests.  A node
failure on a box uploading $us-1$ cache stripes results in $us-1$
cache stripe requests. Each of them may incur a seed request. The
worst event is a video zapping which is equivalent to both events at
the same time\footnote{In the video zapping event from video $v$ to
  video $v'$, the box can can indeed continue to upload the data of
  $v$, but it cannot continue to download more data. In the worst
  case, the buffered data may be scarce for all boxes downloading from
  the box.%
}.  
$r$ adversarial requests thus result in $(u+1)sr$ seed stripe
searches at most.

\begin{claim}\label{nbrecoseed}
All seed stripe searches succeed with probability 
greater than $1-O(\frac{1}{n})$.
\end{claim}

\begin{proof}
We take the point of view of the seed scheduler: a unitary box
is free if its real box is active, and the cache scheduler is
not using it. As the adversary and the cache scheduler operate
independently from stripe allocation, we make the analysis as
if the random choices used for stripe allocation were discovered
as seed requests arrive.
In our case, the purely random scheme consists in allocating
each replica  in a unitary box
chosen uniformly at random.
We show that for $k=O(\log n)$, each replica is considered at most
once with high probability.

For instance, consider a seed request for a stripe $i$.  Its list of
allocated replicas is scanned forward. Each stripe replica falling
in an occupied box is discarded until a replica falls in a free
unitary box.  As observed before, the set $X$ of unitary boxes that
are either under failure or playback cache forwarding is chosen
independently from the replica position.  The set $Y$ of unitary
boxes uploading seeding stripes depends from independent choices for
other stripe replicas.  The probability $p$ that a replica of stripe
$i$ falls in one of the $t=|X\cup Y|$ occupied unitary boxes is
$p=\frac{|X\cup Y|}{usn}$.  Considering that the number of active
boxes is $n_a\ge an$ and that average upload of active boxes remains
$u$ at least, we obtain that the number of failed unitary boxes is
at most $usn - usn_a$. As
the number of current seed connections is
$|Y|$, the number of cache connections is at most $sn_a - |Y|$. We
thus have $|X|\le u(n-n_a)s +sn_a - |Y|$ and $|X\cup Y|\le u(n-n_a)s
+ sn_a\le usn(1 - (1-\frac{1}{u})a)$. We thus have $p\le 1 -
(1-\frac{1}{u})a$.

% LV: mais il peut y avoir peu de seed stripes!!!
% Par Yacine:
% According to our stress-less events hypothesis,
% node failures occur with probability $p_f\leq 1/v_s$.
% Moreover, as discussed previously, for non failed node, the reservation of one stripe for seed
% connections in real boxes ensures that
% a given seed connection is discarded with probability at most
% $\frac{1}{2}$. Combining these two possible events, a seed connection fails with probability
% at most $p_f+\frac{1}{2}(1-p_f)< \frac{v_s+1}{2 v_s}$.

% To determine the maximum number of replicas of the same stripe that are discarded, one can see
% the allocation of slots to these stripes as throwing independantly $n \frac{v_s+1}{2 v_s}$ balls (failed connections)  into $n$ bins (stripes). A standard result for the occupancy problem gives $ \frac{v_s+1}{2 v_s}+log(n)=O(log(n))$ as the maximum number of discarded stripes.

As $r=O(n)$, the number of stripe requests is
at most $\lambda n$ for some $\lambda > 0$.
As discussed previously, the reservation of one stripe for seed
connections in real boxes ensures that
a given seed connection is discarded with probability at most
$\frac{1}{2}$. 
A given stripe is thus discarded
at most $\log_2 \lambda n^2$ times by the cache scheduler with
probability $1-\frac{1}{\lambda n^2}$ at least. 
Similarly, for a given stripe, the event that
a box uploading it with a seed connection fails happens
at most $O(\log n)$ times with high probability
at least according to our stress-less events hypothesis.
Every stripe is thus
discarded at most $O(\log n)$ times with high  probability.
 %at least $1-\frac{1}{n}$.
There may be up to $v_S$ seed connections for a given stripe
and stress-less events start at most $\lambda'\log n$ swarms
on the video of the stripe for some $\lambda'>0$. This results
in $v_S\lambda'\log n$ stripe searches at most. %($v_S$ per swarm start).
Finally, we note that with high probability, every stripe is searched
at most $\lambda''\log n$ times with high probability for some
constant $\lambda''>0$.

The list of replicas of a stripe can thus be seen as a sequence
of zeros (when the replica falls in an occupied unitary box) and
ones (when the replica is found). A zero occurs with probability
less than $p$ and a one with probability more than $1-p$.
We can conclude the proof if the list of ones in all stripe lists
is greater than $\lambda''\log n$ with
high probability.
As random choices for each replica are independent, we conclude
using Chernoff's upper bound that a sequence of
$k=O(\log n / (1 - p))$ replicas contains the required number of ones
with high probability. (Including the parameters of the model, we use
$k=O(\frac{v_S}{a}\frac{u}{u-1}\log n)$.)
\end{proof}

Of course, this vision of consuming the list of replicas of a stripe
is particular to our proof. In practice, one can loop back to  the
beginning of the list when the end is reached.

\jump  %------------- jump this ------------------------
For instance, consider a seed request for stripe a stripe.
Its list of allocated replicas is scanned forward. Each stripe replica
falling in an occupied box is discarded until a replica
falls in a free unitary box. For the $i$th request, we thus
scan the replicas $r_i, r_i+1,\ldots$
where $r_i-1$ denotes the overall number of stripe replicas considered 
for previous requests. According to the random permutation $\pi$
of the allocation scheme, they fall in 
storage slots $\pi(r_i), \pi(r_i+1),\ldots$.
We are interested in bounding the probability that these slots 
belong to occupied unitary boxes.
A uniform way of drawing a random permutation is to pick uniformly
at random a storage slot different from the slots where the
$r_i-1$ first replicas have fallen.  As observed before,
the set $X$ of unitary boxes
that are either under failure or playback cache forwarding
is chosen independently from the permutation. 
The set $Y$ of unitary boxes uploading seeding stripes depends
from the already considered entries of the permutation.
The probability $p_j$ that replica number $j+1$ falls in one of the $t=|X\cup Y|\frac{d}{u}$ storage
slots of boxes is have $p_j=\frac{|X\cup Y|\frac{d}{u} - j}{dns - j}\leq \frac{|X\cup Y|}{uns}$. 

Considering that the number of active boxes is $n_a\ge an$
and the number of current seed connections is $|Y|$, the number
of cache connections is at most $sn_a - |Y|$. We thus
have $|X|\le u(n-n_a)s +sn_a - |Y|$ and 
$|X\cup Y|\le u(n-n_a)s + sn_a\le nsu(1 - (1-\frac{1}{u})a)$. We have then $p_j\le 1 - (1-\frac{1}{u})a$. The probability that all of the $k$ copies of some stripe fall all into the $X\cup Y$ boxes is at most for $k>\frac{2\log(n)}{ (1-\frac{1}{u})a}$, then $p_j^k \leq \frac{1}{n^2} $.

\finjump %------------- end jump ------------------------

\begin{claim}\label{nbrecocache}
All cache stripe searches succeed with probability 
greater than $1-O(\frac{1}{n})$.
\end{claim}

\begin{proof}
We assume a choice of $s$ such that $u \ge \mun + \frac{1}{s}$. As
in Lemma~\ref{lem:reduction}, we suppose that a box $i$ with poor
upload capacity $u_i < \mun + \frac{1}{s}$ reserves 
an upload $\mun + \frac{1}{s} - u_i$ on some richer boxes.
(Note that this augments the probability of failure for the node, 
a problem we do not try analyze here). A rich box forwarding $i$ stripes
to a poor box accepts preferentially connections for these stripes
(as for stripes of the video it is playing) up to an upload bandwidth of
$\mun i$.

First consider the case where a box $b$ is entering the swarm
(i.e. it requests position 0 in the stripe file).
The swarm $Z$ of $v$ can be decomposed in the set $X$
of boxes arrived in $Z$ before time $t-t_S$ and the set $Y$ of boxes arrived in
$Z$ later on. We thus have $Z=X\uplus Y$ and $b\in Y$.
If $X=\emptyset$ then $|Y|\le v_S$ according to the arrival bound
of our model and each seed stripe search succeeds
with high probability as discussed above.
On the other hand, if $X\not=\emptyset$, we have
$|Z|\le \mun |X|$ according
to the exponential bound on swarm churn in our model.
The boxes in $X$ have overall upload capacity $(u-\frac{1}{s})|X|$
(including the capacity reserved on richer boxes)
and serve at most $|Z|-\frac{1}{s}$ times the video
(box $b$ is still searching for a stripe).
As $u-\frac{1}{s}\ge\mun$, some connection slot is free for
accepting the stripe connection of box $b$. It can always be found
if $b$ has the full list of boxes in the swarm.
The fraction of boxes with exceed capacity for the video
is thus at most $\frac{\mun}{u-1/s}$. 
Note that a slightly higher value of $u\ge \mun+\frac{2}{s}$
would result in a constant fraction of nodes with exceeded capacity
for their video. This would allow to find one
with high probability if the list of random nodes in the
swarm has length $O(\log n)$.

Now consider the case where a box $b$ is reconnecting in its swarm due
to some zapping or node failure event. We can prove similarly that
the connection flipping algorithm allows to find a node in the swarm
to connect to. This relies on the hypothesis that the number of
reconnecting nodes
at position $t$ in a video increases by a factor $\mun$ at most
during a period of time $t_S$ as assumed by our model
(all types of swarm churn are aggregated in the bound $\mun$).
\end{proof}

\jump %------------------------ jump ---------------------------------
\bigskip
\subsection{tentative de preuve par FdM}
Four kinds of boxes
\begin{enumerate}
\item zapped : an event occurred on this box less than $t_S$
  ago. Nature of event : state change, failure,
  resurrection... proportion globale BORNE par $(1-\mun)$. 
  Ont du dl mais pas d'up du  tout (vi on peut le supposer)
\item failed (ni up ni dl) proportion BORNNEE PAR $<(1-a)$
\item active : regardent le film depuis plus de $t_S$. Ont du up et du
  dl. proportion : $\aleph$
\item idle. Ont du up mais pas de dl. Bien sur pour faire chier on
  peut toujours supposer que y'en a pas... Sinon proportion $\iota$
\item reconnected. Sont des boites actives, mais dont la connectivit
  chang au cours de $t_S$. Elles continuent a uploader, grce a leur
  buffer. Mais faut pas que y'en ait trop...
\end{enumerate}

\medskip

Evenements possibles : 
\begin{itemize}
\item zap. Active $\to$ zapped.
\item allumage. Idle $\to$ zapped
\item exctinction OU PANNE: active $\to$ zapped
\item Fin de startup zapped $\to$ active
\item Fin de cache : zapped $\to$ idle
\item Panne averre : zapped $\to$ failed
\end{itemize}

\medskip

Upload total : $(\aleph+\iota)u$ au moins\\
Download total $\aleph + (1-\mun)$ au plus\\
Donc pas de pb sur les bandes passantes globales meme si $\iota=0$
parce qu'on a justement $u>\mun$.

\medskip

On rintroduit les boites full : sont actives ou idle, mais a bloc

\medskip

\textbf{claim debloquant}
 le nombre de copies dispos (i.e. dans une boite qui peut les
 uploader) est $k'=k/constante$ whp.
\hrule
Demo : il y a une proportion constante de boites
\begin{itemize}
\item eteintes : $(1-a)$
\item zappees : $(1-\mun)$
\item full : $1/u$
\end{itemize}
Soit $z = (1-a)+(1-\mun)+1/u$ un majorant de cette proportion. Of
course faut que ce soit $<1$\\
$yzk$ doit etre un majorant du nombre de boites avant de l'upload
dispo, o $y$ est obtenu quand on est bon en proba et est fonction de
la proba cherche.

Finalement ne nombre cherch est $k' yzk - \kappa$ car il faut enlever
les au plus $\kappa$ connections store $\to$ swarm.

Si on s'est pas gour, $k'$ est plus que 1.
\hrule

\medskip
Point~2 : du coup ni le seed ni le cache scheduler ne peuvent faillir
! C'est comme dans l'ancienne preuve. La proba est donc dans le fait
qu'il doit tout le temps rester au moins une copie uploadable par
stripe. Qu'en pensez-vous ?

\laurent{OLD STUFF bellow}

\begin{theorem}
\label{th:realistic_old}
Every proportionally heterogeneous  $(n,u,d)$-video system
with $u\ge 1+\frac{1}{c}$ can achieve catalog size
$\Omega(\frac{dn}{\log_{u/a} ns})$ with high probability

associated to the random connection scheduler
can satisfy any sequence of node state changes with high probability
in the realistic adversary model with realistic exponential arrivals and with at least $a$ active nodes.
\end{theorem}
% If we suppose $a$ is large (near to $1$) and $u\ge 1+2/s$ then the
% lower bound on catalog size approximately becomes $\frac{2dn}{2s\log
%  n + s\log s}$, i.e. is scalable with a factor $\frac{1}{\log n}$.
%The proof of this theorem in is Section~\ref{sec:preal}

\subsection{Preference to relay}

Indeed, for $u\ge 1+1/s$, we
can select a \emph{preemptive} bandwidth $u'\ge \frac{\mun}{s}$.
More
generally each box $i$ has a preemptive bandwidth threshold
 such that $u_i>u'_i\ge \frac{\mun}{s}$.
We suppose that the lookup mechanism includes a counter incremented
each time a box enters the swarm for a video $v$. 
That way, each box can select a preferential stripe $x$ of $v$
in a cyclic manner.

If Box $i$ receives a request and has some free upload capacity, it accepts.
%for stripe $y$ and already uploads it
%it refuses. A stripe that is not viewed is uploaded at 

If Box $i$ already uploading at its maximum rate, it may accept a new
connection from box $j$ (and therefore reject an old one $j'$) in the
following case: $i$ is watching a video $v$ and has preferential
stripe $x$, $j$ is requesting $x$,
$i$ uploads less than $u_i's$ times stripe of $x$,
and $j'$ is not watching $x$. The \emph{preemptive} bandwidth,
i.e. the bandwidth of each box $i$ allocated for the 
\emph{preferential stripe relaying}
rule is $u'_i< u_i$ (strictly less : there is at least a non-preemptive
upload slot). 

This schemes implies that a downloading box may be
\emph{disconnected}, and has to find a new uploading box. Fortunately,
as shown below, under good hypotheses a disconnected box can always
reconnect using the ``non-preemptive'' bandwidth of another box (the
remaining $u_i-u_i'$).

\begin{claim}\label{nbreco}
Each disconnected box can reconnect to another box without creating a new disconnection.
\end{claim}

A sketch of proof is for the claim and the
theorem~Theorem~\ref{th:realistic} is given
in Section~\ref{sec:preal}.

\subsection{Nodes failures: Managing the proportional heterogeneity}\label{sectrealhetero}
The reduction from proportional heterogeneity to homogeneity is done
using Lemma~\ref{lem:reduction} construction: %\fabien{verifier}
each \emph{poor} node (with upload less than 1) receives bonus slots
from rich nodes (with upload $\ge 1+\frac{1}{s}$) acting as 
proxies and caching the stripe forwarded.

If a \emph{poor} node fails, this is no problem for the system (if
the ratio stays under $1-a$ of course, see proof in
Section~\ref{sec:preal}). But if a rich node fails, the poor nodes
that use it as proxy must find another proxy. If the bandwidth is
high enough, such proxy always exists. A scheme enabling dynamic
upload slot renting has to be used. We can image a similar
lookup procedure as for finding boxes with a given stripe.
% But during the time $t_S$
% they need to reconnect, their own upload risks to be cut (during
% $t_S$ at most, less with buffering). So failures creates exactly the
% same delays than connection of new nodes, namely $t_S$. That's why
% exponential arrivals assumption counts failures just like arrivals.
% 
So we are lead to Homogeneity case. The proof of
Theorem~\ref{th:realistic} is completed in Section~\ref{sec:preal}.

\subsection{Playback caching index}

As discussed bellow, video allocation updates are made at
long time scale with respect to playback time.
The list of boxes storing a given video according to 
the video allocation scheme can thus be directly associated to the
video in the menus proposed to the user when browsing the catalog.
We may thus assume that this information available to
the box as soon as the user launches the playback.

Practical scheduling implementation relies on a distributed hash
table (DHT). The DHT stores, for each stripe, the set of
\emph{available} boxes, i.e. the box who can upload that stripe.
\laurent{Question de comment l'utilisateur browse le catalog.}

Every available box registers in the DHT and give its upload
power.  When a box wants to download, a DHT lookup is performed to
obtain a list of candidate goods uploaders, and it can contact them
and begin downloading to the best one (will largest remaining
bandwidth). As shown below, it is useful that boxes storing
due to the video allocation scheme and the boxes storing
due to playback caching register separately.

\subsection{Video allocation updates at long time scale}

On a long time scale, new boxes enter the system and old boxes
leave. If the system is stable, video stripes disappearing
should be re-allocated in new boxes. If the system grows,
new videos may be added to the catalog. 
The random allocation scheme may be incremental:
select randomly some boxes and put their stripes in the new boxes.
The space liberated is filled up with the copies of the stripes of the
new videos.
\laurent{Yacine est-ce que toutes les permutations sont
toujours equiprobables ?}

Notice that the caching of the last viewed video, used for assisting
massively viewed videos, could be extended to some kind of LRU storage
system: the $k'$ most recently viewed video are stored in the
boxes. Then the number of replicates of a video is no more a constant
$k$ but follows video popularity. If the storage space is decreased
for increasing the cache space, the bound on the catalog size decrease
(Theorem~\ref{th:realistic}) so this is useless. Furthermore, we have seen
that the \emph{preference to the relay} rule ensures that popular
videos request are \emph{less} subject to be rejected than rare video
requests. On the other hand, the dynamics of the scheduler are easier
to handle (as the DHT is less queried) and the system is more robust
with respect to targeted attacks (against boxes storing popular
videos) from a \emph{strong} adversary.

\finjump %------------------------ end of jump -----------------------

% \begin{figure}[t]
% \noindent\hspace{-.9cm}\includegraphics[width=.60\textwidth]{../figures_sigcomm/multiple_vs_k} 
% \caption{Requests satisfied as a function of $k$.}
% \label{fig:multiple_vs_k}
% \end{figure}

\section{Simulations}
\label{sec:simus}
In this section we evaluate the performance of a practical
allocation scheme by the dint of simulations. This scheme is similar
to the one described in section ~\ref{sec:realistic} but it presents
two main differences. Firstly, the storage allocation is based on a
random regular graph obtained by a permutation $\pi$ of the $kms$
stripes into the $dns$ storage slots. This choice is motivated by
the more practical aspect of regular random allocation that allows
to completely fill-in boxes.
Secondly, once the connections
are established, they cannot be re-negotiated when a new
video-request is performed. The goal is to test the basic
functionning of the algorithm to understand where connection
renegociations become necessary.

We assume that every node has a cache of size $1$ where it stores all
the stripes of the video it is watching. We suppose video requests arrive at a constant rate, and $t_S=2$ minutes.

As stated in previous sections, the efficiency of an allocation
scheme depends on the requests pattern. In the following, we use
five kind of 
adversarial schedulers to generate video request sequences:
\begin{itemize}
\item {\bf Greedy adversarial}. The  greedy adversarial
  scheduler chooses the request for which the system will select a
  node with minimal remaining upload bandwidth (among the set of
  nodes that can be selected by a request in the current
  configuration). This adversary make greedy decisions. It is strong
  in the sense that it is aware of video allocation and current connections. 
%   is greedy and unaware of the exact stripe
%   allocation. 
  % FM : suite inutile ? Moreover, because of the cache,
              % the bandwidth available for a selected video is
              % increased.
\item {\bf Random}. The random scheduler selects a video uniformly at random in the catalog.
\item {\bf Netflix}. $m$ videos are randomly selected from the
  \emph{Netflix Prize} dataset \cite{netflix} as catalog for our
  simulated system. Requests are performed following the real
  popularity distribution observed in the dataset.
\item {\bf Netflix2}. The $m$ most popular videos of the Netflix
  Prize dataset are selected as catalog for our simulated
  system. Requests are performed following the real popularity distribution of these $m$ videos.
\item {\bf Zipf}. The scheduler selects videos following a Zipf's law popularity distribution with $\gamma=2$.
\end{itemize} 
The peers that perform a request follow a sequence of random permutations of the $n$ peers. All our simulations are performed with $n=100$ nodes, and the results are averaged over multiple runs.

\subsection{Impact of the number of copies per video}
We study the maximum number of requests the system is able to satisfy as a function of the number $k$ of copies per video ($k\approx \frac{nd}{m}$). We suppose that nodes may watch more than one video (for instance if multiple playback devices depend on a single box) so the total number of requests can be larger than $n$, even if $n$ is the typical desired target. We set $s=15$, $u=1+\frac{1}{s}$ and $d=32$.
Figure~\ref{fig:multiple_vs_k} shows that the system is able to satisfy at least one request per node if $k\ge6$, independently from the requests pattern. Moreover, for the Random, Netflix and Netflix2 schedulers, $k\ge3$ is enough.

\newcommand{\figfig}[3]{
\begin{figure}[ht]
\begin{center}
\scalebox{0.8}{
 \includegraphics[width=.8\textwidth]{#1}
}
\end{center}
\vspace{-1cm}
\caption{#3}
\label{fig:#2}
\end{figure} 
}

\figfig{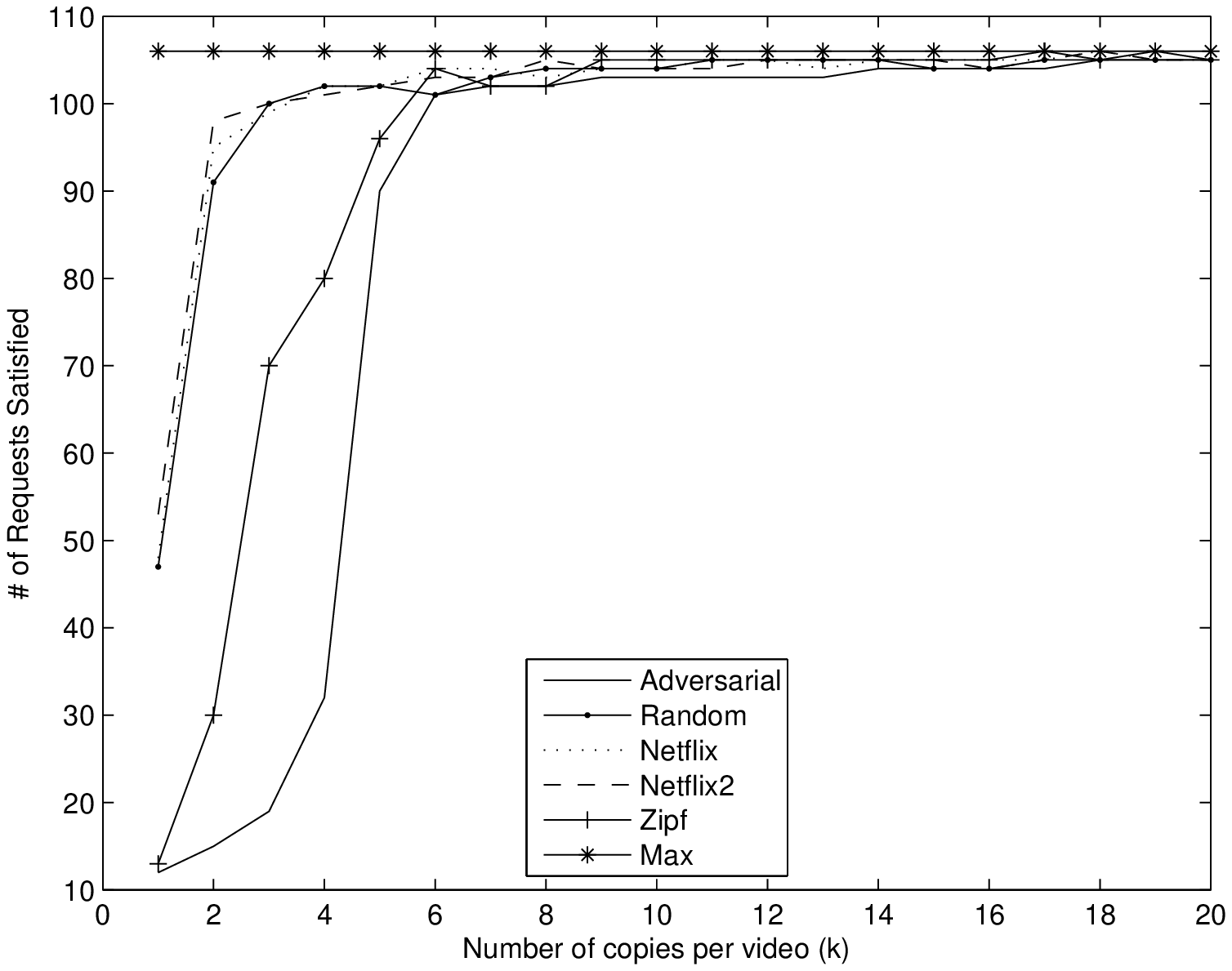}{multiple_vs_k}{Requests satisfied as a function of $k$. $n=100$, $d=32$, $s=15$, $u=1+\frac{1}{s}$}

We indicate as reference the maximum number of requests the system can satisfy considering the global available upload bandwidth. Note, that for $k\ge10$, nodes almost fully utilize their upload bandwidth and the system asymptotically attains the maximum possible number of requests.

\subsection{Varying the number of stripes}
We study the impact of the number of stripes into which videos are split. For this purpose, we set $k=10$ , $d=32$ and $u=1+\frac{1}{s}$.

\figfig{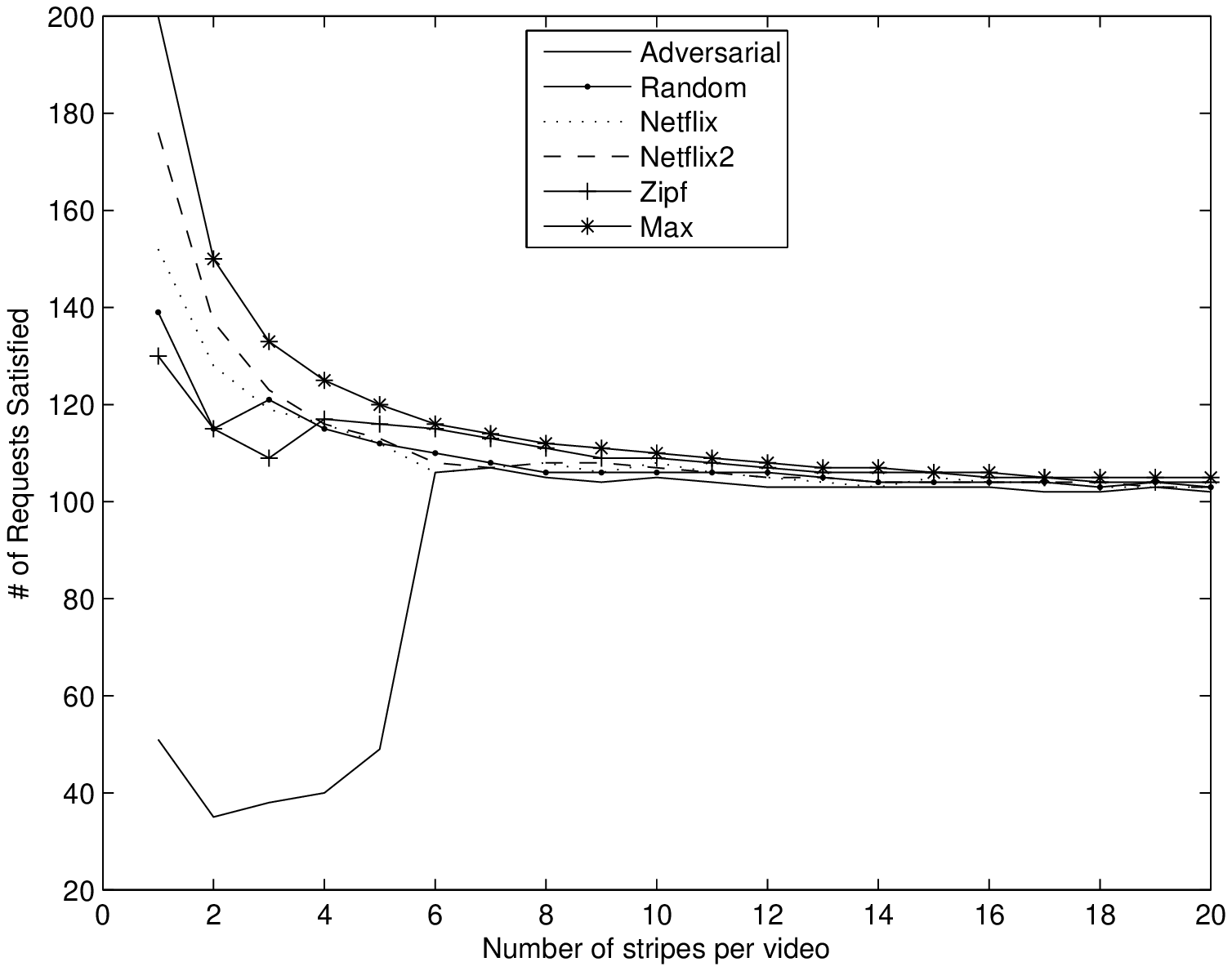}{multiple_vs_S}{Requests satisfied as a function of $s$. $n=100$, $d=32$, $k=10$, $u=1+\frac{1}{s}$}

Figure~\ref{fig:multiple_vs_S} shows that the system can satisfy $n$
requests or more for all schedulers but the adversarial.
With few stripes, the greedy scheduler may find blocking situations
were re-configuration of connections would indeed be necessary.
For low $s$, more requests are served with other schedulers. This is not
surprising, considering that a reduction of the number of stripes
leads to an increase of the system global bandwidth. As $s$
increases, $u$ tends toward $1$ and the number of satisfied requests
to $n$.

% \begin{figure}[t]
% \noindent\hspace{-.9cm}\includegraphics[width=.60\textwidth]{../figures_sigcomm/multiple_vs_S} 
% \caption{Requests satisfied as a function of $s$.}
% \label{fig:multiple_vs_S}
% \end{figure}

%\pagebreak
\subsection{Heterogeneous capacities}
We analyze the impact on the number of video requests satisfied in
presence of nodes with different upload capacities. Node capacity
distribution is a bounded Gaussian distribution with
$u=1+\frac{1}{s}$ and different variance values.  We set
$k=10$, $s=15$ and $d=32$.  Figure~\ref{fig:heteroigeneity} shows
the results. Schedulers can satisfy at least $n$ requests for small
or large values of upload variance, with a slight loss of efficiency
between. This may  come from the fact that we do not use a
proportional allocation scheme here.

\figfig{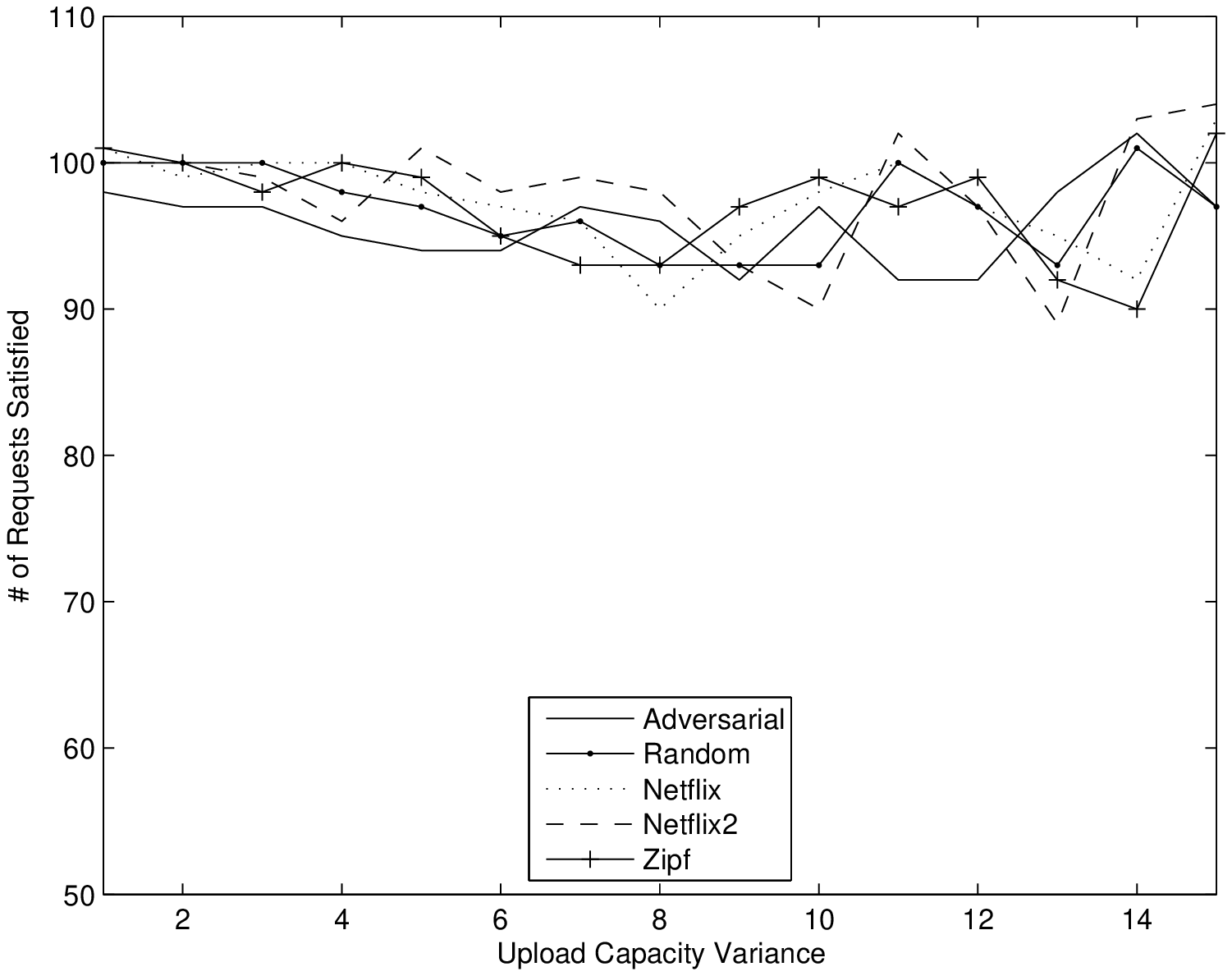}{heteroigeneity}{Requests satisfied with heterogeneous capacities. $n=100$, $d=32$, $k=10$, $s=15$, $\bar{u}=1+\frac{1}{s}$}

% \begin{figure}[t]
% \noindent\hspace{-.9cm}\includegraphics[width=.60\textwidth]{../figures_sigcomm/heterogeneity} 
% \caption{Requests satisfied with heterogeneous capacities.}
% \label{fig:heteroigeneity}
% \end{figure}

\subsection{Node failures}

We evaluate the impact of off-line peers on the number of video requests the system can satisfy. We set $k=10$, $s=15$, $d=32$ and $u=1+\frac{1}{s}$. We then randomly select some nodes and we set them inactive for the simulation.

\figfig{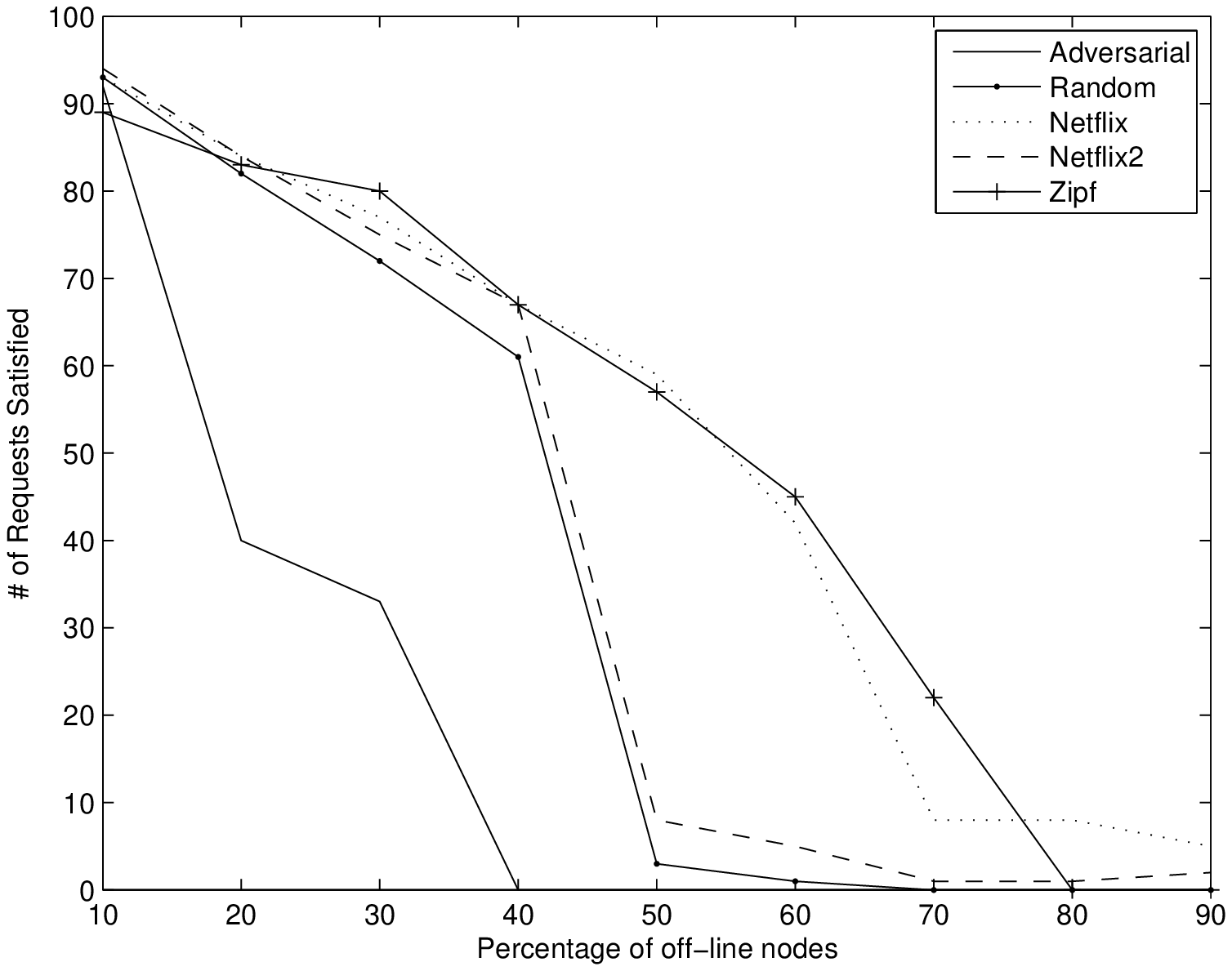}{static-offline}{Number of requests satisfied with static off-line peers. $n=100$, $d=32$, $k=10$, $s=15$, $u=1+\frac{1}{s}$} 

Figure~\ref{fig:static-offline} shows the system can satisfy video
requests for at least all the active nodes in the system up to 40\%
failures $(a=0.6$). Then, a drastic decrease in the performance
occurs. As soon as there are 10\% of boxes off-line, the adversarial
scheduler is able to block the system.
% Only the adversarial scheduler is not able to satisfy video
% requests when more than 10\% of boxes are off-line.

% \begin{figure}[t]
% \noindent\hspace{-.9cm}\includegraphics[width=.60\textwidth]{../figures_sigcomm/peer_departure_static} 
% \caption{Number of requests satisfied with static off-line peers.$n=300$, $d=32$, $k=10$, $s=16$,$u=1+\frac{1}{s}$}
% \label{fig:static-offline}
% \end{figure}

\section{Conclusion}
\label{sect:conclusion}
%It is better to have more boxes or more bandwidth than more disk space.

In this paper, we show an average upload bandwidth threshold
for enabling a scalable fully distributed video-on-demand system.
Under that threshold, scalable catalog  cannot be achieved.
Above the threshold, linear catalog size is then possible
and the problem of connecting nodes to serve demands reduces to
a maximum flow problem. A slight upload provisioning allows to
build distributed algorithms achieving scalability.

% that a fully scalable VoD catalog can be
% designed. However, subtle trade-offs must be made between bandwidth,
% storage capacity and failure rates. One should retain that when the
% available bandwidth is too critical, or when byzantine failures
% occur, the catalog is sparse. Things get better for reasonable
% failures and slightly over-provisioned bandwidth, allowing
% proportional (in theory) or almost proportional (in practice)
% catalog scaling.

%\nocite{*}
\bibliographystyle{plain}
\bibliography{podc-short}

\newpage 

\section*{Appendix A}

\subsection*{Maximum flow scheduler}
\label{proof:expander}

We prove Theorem~\ref{th:expander} thanks to the two following
lemmas. For the sake of clarity, the proof  is written for
the homogeneous case. It is discussed later on how it generalizes
to heterogeneous capacities.

\begin{lemma}[Min-cut max-flow]
\label{lem:hall}
Consider a bipartite graph from $U$ to $V$ and an integer $b>0$.
There exist a $b$-matching where each node of node of $U$ has degree
1 and each node of $V$ has degree at most $b$ iff each subset
$U'\subseteq U$ has at least $|U'|/b$ neighbors in $V$
(i.e., the graph is a $1/b$-expander).
\end{lemma}

\begin{proof}
The $1/b$-expander property is clearly necessary. We prove it is
sufficient by considering
the flow network obtained by adding a source node $a$ and
a sink node $z$
to the bipartite graph. An edge with capacity 1 is added from $a$
to each node in $U$. Edges of the bipartite graph are directed from
$U$ to $V$ and have capacity 1. An edge with capacity $b$ is added
from each node in $V$ to $z$. The $1/b$-expander property
implies that every cut has capacity $|U|$ at least.
The well-known min-cut max-flow theorem allows to conclude.
%~\cite{fordfulkerson}. ho pas besoin de ref pour ca ici ! FdM
\end{proof}

\begin{lemma}
\label{lem:rand}
Consider a random regular permutation graph of $kms=dns$ stripe
copies into the $dns$ memory slots of $n$ boxes.
The probability that $ki$ given copies fall into $p$ given boxes
with $pds\ge ki$ is less than $\left(\frac{p}{n}\right)^{ki}$.
\end{lemma}

\begin{proof}
Drawing uniformly at random a permutation of the $kms=dns$ stripes amounts to choose uniformly at random a slot for the first stripe, then a slot for the second among the remaining slots and so on. The $ki$ stripes are ordered. Let $E_a$ denotes the event that the $a^{th}$ copy of stripe falls into one of the $pds$ slots of the $p$ boxes.
$P(\cap_{a \leq ki}E_a)=$ $P(E_1).P(E_2|E_1)...P(E_a|E_1 \cap E_2...\cap E_{a-1})...=$$\frac{pds}{nds}.\frac{pds-1}{nds-1}...$ $\frac{pds-a+1}{nds-a+1}...\leq \left(\frac{p}{n}\right)^{ki}$ (since $\frac{pds-i}{nds-i}\leq \frac{pds}{nds}$ for $p\leq n$).
%Consider permutations of the $kms=dns$ stripes.
%The number of permutations such that the $ki$ copies fall into
%the $pds$ memory slots of the given boxes is
%$\binom{pds}{ki}(ki)!(dns-ki)!$. As there are $(nds)!$
%permutations, the probability is 
%$\binom{pds}{ki}(ki)!(nds-ki)!/(nds)!
% = \frac{(pds)!/(pds-ki)!}{(nds)!/(nds-ki)!}
% = \frac{\binom{pds}{ki}}{\binom{nds}{ki}}$.
%Using the classical bounds
%$\binom{x}{y}\le \left(\frac{xe}{y}\right)^y$ 
%and $\binom{x}{y}\ge \left(\frac{x}{y}\right)^y$, 
%we get the desired bound.
\end{proof}

%We now prove Theorem~\ref{th:expander} in the homogeneous case.

\begin{proof}[of Theorem~\ref{th:expander}]
We assume that $s\le c$ is sufficiently large 
to ensure $u\ge 1+\frac{1}{s}$. We suppose   $u\ge \mun$ and
$s\ge 2$.
Consider the multiset of stripe requests at some time $t$.
Its size is $ns$ at most as there are no more than $n$ videos played.
Let $S$ be a sub-multiset
of size $i$ among the requested stripes. Let $i_1$ be the
number of pairwise distinct requests in $S$ and $i_2=i-i_1$
be the number of duplicated requests in $S$.
As swarm growth is bounded by $\mun$, there are at least
$\alpha i_2$ nodes where duplicate request can be downloaded 
with $\alpha=\frac{1}{\mun}$.

Let $B(S)$ denote the set of boxes from which any stripe of $S$ may
be downloaded. 
From Lemma~\ref{lem:hall}, a connection matching for serving the
request can always be found if no multiset $S$ of at most $rs$
requested stripes verifies $|B(S)|<j$ with $j=\frac{i}{us}$.
Note that $B(S)$ includes at least the 
given boxes where duplicate requests may be downloaded
thanks to playback caching. This represents at least $\alpha i_2$
boxes and $|B(S)|\ge j$ for  $\alpha i_2\ge j$. We may thus
consider only $\alpha i_2<i/us$ (implying $i_1>(1-1/\alpha us)i$).
By summing over all sets of $j=i/us$ boxes
and using Lemma~\ref{lem:rand}, we get the following bound
relying only on the stripe copies placed according to the video
allocation graph (this probability is $0$ for $i \leq us$):
$P(|B(S)|<j)\le \binom{n}{j}\left(\frac{j}{n}\right)^{ki_1}
  \le \left(\frac{unse}{i}\right)^{i/us}
      \left(\frac{i}{uns}\right)^{ki_1}
$. The last inequality is obtained by using the standard upper bound of the binomial coefficient ($\binom{b}{a} \leq \left(\frac{b e}{a}\right)^a$).

Using Markov inequality, the probability that some obstruction
multiset $S$ for some request exists is bounded by the expected number of
such obstructions. By summing the above inequality 
over all multisets $S$ of at most
$ns$ stripes, we get the 
following bound on the
probability $p$ that the graph cannot satisfy all possible requests:
$p \leq \sum_{i=us}^{ns} \sum_{i_1=i(1-1/\alpha us)}^i 
     M(i,i_1)
       \left(\frac{unse}{i}\right)^{i/us}
       \left(\frac{i}{uns}\right)^{ki_1} $
where $M(i,i_1)$ is the number of multisets of cardinality $i$ taken
from sets of stripes of cardinality $i_1$. $M(i,i_1)$ is at most
$M(i,i_1) \leq \binom{\lfloor nds/k \rfloor}{i_1}
\binom{i+i_1-1}{i_1-1} \leq \left(\frac{ndse}{ki}\right)^i
\binom{2i}{i}\leq \left(\frac{4ndse}{i}\right)^i$ since $i_1 \leq i$
and considering that $k\geq 1$. Notice also that
$\left(\frac{i}{uns}\right)^{ki_1} \leq
\left(\frac{i}{uns}\right)^{ki-ki/\alpha us} $. 
The probability is then at most:
$p \leq
 \sum_{i=us}^{ns} \frac{i}{\alpha us} \left(\frac{i}{uns}\right)^{\kappa i} 
   \delta^i \leq \frac{n}{\alpha u}\sum_{i=us}^{ns}
   \left(\frac{i}{uns}\right)^{\kappa i}\delta^i $ where
   $\delta=4de^{1+1/us}/u$ and $\kappa=k-k/\alpha us-1/us-1$. 

%For
 %  the sake of simplicity, we take practically reasonable values for
 %  $s$ and $u$, so that $us\geq \alpha us \geq 2$, we can write then $\kappa \geq k/2-3/2$. Similarly, by taking a capacity $d \geq 2e^{1+1/us}/u$, we have $\delta \leq d^2$. 

%With these assumptions, the probability is bounded by $p \leq
%\frac{n}{\alpha u}\sum_{i=us}^{ns} \left(\frac{i}{uns}\right)^{(k/2-3/2) i}d^{2i} $.

It is easy to check that as a function of $i$ the terms of the sum
$\phi(i)=\left(\frac{i}{uns}\right)^{\kappa} \delta^{i}$ decrease
from $\phi(us)$, reach a minimum at $\phi(i^\star)=\phi\left(\frac{uns}{\delta^{1/\kappa}e}\right)$ then increase to $\phi(ns)$. Using this
fact, we bound $p$ by considering separately the sum for $i<i^\star$
and $i>i^\star$ and by replacing each term with the maximum term on
its side. On one hand, $\frac{n}{\alpha u}\sum_{i=us}^{\lfloor i^\star
  \rfloor} \left(\frac{i}{uns}\right)^{\kappa} \delta^{i}$ 
$\leq \frac{n}{\alpha u}.ns.\phi(us)
=\frac{n}{\alpha u}.ns.\frac{1}{n^{\kappa us}} \delta^{us}$
$ \leq O\left(\frac{1}{n^{\kappa us}}\right)$. On the other hand,
the sum of the terms of rank greater than $i^\star$ gives
 $ \frac{
  n}{\alpha u} \sum_{\lfloor i^\star \rfloor+1}^{ns}
\left(\frac{i}{uns}\right)^{-\kappa i}\delta^{i}$
$\leq \frac{n}{\alpha u}.ns. u^{-\kappa n s}\delta^{ns}$
$ \leq O\left(n^2
  \left(u^{-\kappa}\delta\right)^{ns}\right)$. Finally,
  $p\leq  O\left(\frac{1}{n^{\kappa us}}\right)+O\left(\left(u^{-\kappa}\delta\right)^{ns}\right)$. For the first term to vanish, we need $u^{-\kappa}\delta<1$ and then $\kappa>\log_u(\delta)$. For this, we need to replicate each stripe at least $k$ is then $k>\log_u(d)\frac{\alpha us}{\alpha us-1}+\frac{\alpha +\alpha us \log_u(4 e^2)}{\alpha us-1}$. For the sake of simplicity, consider $s \geq 2$ the lower bound on the number of replicates is $k>2 \log_u(d)+2 \log_u(4 e^2)+1$. In this case the probability of failure is at most $p \leq O\left(\frac{1}{n^{\kappa us}}\right)$ (note that $\kappa us > 0$) and then the bipartite graph can satisfy all possible requests with high probability. Since the number of videos that can be stored is $nd/k$ and given the condition on $k$, the storage capacity is $\Omega(nd/\log_u(d))$.

\end{proof} 

Now consider the heterogeneous case.
Lemma~\ref{lem:hall} and the above proof may be generalized.
Recall that box $b$ has storage
capacity $d_b$ and upload capacity
$u_b$.  The condition for an obstruction then
becomes $\sum_{b\in B(S)} su_b< |S|=i$.
We can then consider any subset $E$ of boxes with overall capacities
$U_E=\sum_{b\in E} u_b$ and $D_E=\sum_{b\in E}d_b$ such that
$U_E < i/s$. 
As boxes are chosen according to their capacity,
the probability to put a stripe in $E$ with random allocation is thus 
$\frac{D_E}{nd} = \frac{D_E U_E}{ndU_E}<\frac{D_E}{dU_E}\frac{i}{ns}$
for an obstruction.
Assuming $d\frac{U_E}{D_E}\geq u'$ for all $E$, we can follow the
same tracks for the proof with the probability of obstruction being
less than $\frac{j}{n}$ with $j=\frac{i}{u's}$.
With a smallest upload  capacity of 1 stripe, the total number
of such sets $E$ is bounded by $\binom{usn}{i}$ instead of
$\binom{n}{j}$ in the above proof.
This larger factor in the sum over all multiset of request
is not a problem when taking a slightly larger value of $k$.
We can thus  get similar 
bounds as long as $d\frac{U_E}{D_E}\geq u'$ for some
$u' \ge \max\set{1+\frac{1}{c},\mun}$. 

Boxes with capacity lower
than $u'$ can be grouped with high upload capacity boxes
to obtain the desired property as proposed 
in Lemma~\ref{lem:reduction}.

\subsection*{Poor Upload Capacity Boxes}

\begin{proof}[of Lemma~\ref{lem:reduction}]
Boxes $b$ with
upload capacity $b<\mun$ are said to be \emph{poor}.
Boxes with
upload capacity exactly $\mun$ are said to be \emph{medium}.
Boxes $b$ with  $u_b> \mun$
upload capacity are said to be \emph{rich}. 
Let $P$, $M$ 
and $R$ denote the sets of poor, medium 
and rich boxes respectively. We set $n_P=|P|$, $n_M=|M|$, 
and $n_R=|R|$
(Note that $n=n_P+n_M+n_R$). Let $u_P=\frac{U_P}{n_P}$ 
and $u_R=\frac{U_R}{n_R}$ 
be the mean and overall upload 
capacities of poor and rich
boxes respectively.
 
We construct $B$ from $A$ with same video allocation.
For the sake of simplicity, we assume that $s/\mun$ is an integral
value as well as $u_bs$ for each box $b$.
% Notice that each box output an integer number of stripe, so the upload
% may be considerer as a multiple of $\frac{1}{s}$.
In a pre-processing step, for each poor box $b$, we 
reserve $\mun(1-\frac{u_b}{\mun})s=\mun s - u_b s$
upload slots from the rich boxes. This upload will be used to
forward $(1-\frac{u_b}{\mun})s$ stripes to the poor box and
serve new arrivals in the swarm for up to
$(\mun-1)(1-\frac{u_b}{\mun})s$ stripes.
This assignment should
balance the overall number $s_b$  of slots reserved on a rich box $b$
such that its remaining upload
capacity $u_b'=u_b-\frac{s_b}{s}$ remains no less than $\mun$.
(In a proportionally heterogeneous system, one would typically
choose $s_b$ proportional to $u_b$.)
A corresponding space of $\frac{s_b}{\mun s}$ should also be additionally be 
reserved for playback caching. We thus set $d'_j=d_j+\frac{s_j}{\mun
  s}$
This assignment is possible when $U_R-\mun n_R\ge \mun n_P - U_P$, i.e. $u\ge \mun$.
Now we use a video allocation scheme for capacities $u_b', d_b'$
(where $u_b'=u_b$ and $d_b'=d_b$ for each poor or medium 
box $b$). The connection
scheduler works as previously except for the download connections
of poor boxes. When a poor box $b$ requests a video, the $s$ stripes are
downloaded from the boxes decided by the previous scheme.
However, $b$ downloads $u_b\frac{s}{\mun}$ 
stripes directly but the $(1-\frac{u_b}{\mun})s$
others are downloaded via the rich boxes with reserved upload slots 
for box $b$. These rich boxes participate in the caching
of the stripes they forward instead of $b$. This scheme allows to
increase the overall upload capacity of the set $E$ of all boxes
caching some stripe requested by $p$ boxes so that $U_E\ge \mun p$.
\end{proof}

\end{document}